\documentclass[11pt]{article}
\usepackage{blindtext}
\usepackage{geometry}
\usepackage{amsmath}
\usepackage{amsthm}
\usepackage{hyperref}
\usepackage{amssymb}
\usepackage{caption}
\usepackage{subcaption}
\usepackage{cleveref}
\usepackage{enumitem}
\usepackage{graphicx}
\usepackage{tikz}
\usetikzlibrary{quantikz2}
\usepackage{braket}
\usepackage[colorinlistoftodos]{todonotes}
\usetikzlibrary{positioning, shapes,calc}
\usepackage[makeroom]{cancel}
\usepackage{comment}
\usepackage{hyperref}
\usepackage{marvosym}
\usepackage{fontawesome}

\usepackage{tikz}
\usepackage{lipsum}
\usepackage{float}
\usepackage{braket}
\usepackage{mathtools}

\newtheorem{result}{Result}
\newtheorem{theorem}{Theorem}
\newtheorem{lemma}[theorem]{Lemma}

\newtheorem{corollary}[theorem]{Corollary}

\newtheorem{definition}[theorem]{Definition}
\newtheorem{remark}{Remark}

\newcommand*{\NC}{\mathsf{NC}}

\newcommand*{\AC}{\mathsf{AC}}
\newcommand*{\TC}{\mathsf{TC}}
\newcommand*{\QNC}{\mathsf{QNC}}
\newcommand*{\MOD}{\mathrm{MOD}}
\newcommand*{\QAC}{\mathsf{QAC}}

\newcommand*{\iQNC}{\mathsf{i}\text{-}\mathsf{QNC}}
\newcommand*{\fQNC}{\mathsf{QNC}}

\newcommand*{\iQAC}{\mathsf{i}\text{-}\mathsf{QAC}}
\newcommand*{\fQAC}{\mathsf{QAC}}

\newcommand*{\iQTC}{\mathsf{i}\text{-}\mathsf{QTC}}

\newcommandx{\info}[2][1=]{\todo[linecolor=grey,backgroundcolor=grey!25,bordercolor=grey,#1]{#2}}
\newcommandx{\change}[2][1=]{\todo[linecolor=blue,backgroundcolor=blue!25,bordercolor=blue,#1]{#2}}
\newcommandx{\missing}[2][1=]{\todo[linecolor=red,backgroundcolor=red!25,bordercolor=red,#1]{#2}}
\newcommandx{\agnote}[1]{\todo[linecolor=green,backgroundcolor
=green!25,bordercolor=green]{Alex: #1}}
\newcommandx{\mnote}[1]{\todo[linecolor=blue,backgroundcolor
=blue!25,bordercolor=blue]{Michael: #1}}

\title{The Power of Shallow-depth Toffoli and Qudit Quantum Circuits}

\author{Alex B. Grilo$^{1}$, Elham Kashefi$^{1,2}$, Damian Markham$^{1}$, Michael de Oliveira$^{1,3,}$\thanks{\tiny\faEnvelopeO\ \href{mailto:michaeldeoliveira848@gmail.com}{michaeldeoliveira848@gmail.com}} \ \\
\normalsize
   \textit{$^1$Sorbonne Université, CNRS, LIP6, France} \\ \normalsize 
   \textit{$^2$School of Informatics, University of Edinburgh, Scotland} \\ \normalsize
   \textit{$^3$
International Iberian Nanotechnology Laboratory, Portugal}\\
}

\date{}

\begin{document}

\maketitle

\begin{abstract}
The relevance of shallow-depth quantum circuits has recently increased, mainly due to their applicability to near-term devices. In this context, one of the main goals of quantum circuit complexity is to find problems that can be solved by quantum shallow circuits but require more computational resources classically.

Our first contribution in this work is to prove new separations between classical and quantum constant-depth circuits. Firstly, we show a separation between constant-depth quantum circuits with quantum advice $\QNC^0/\mathsf{qpoly}$, and $\AC^0[p]$, which is the class of classical constant-depth circuits with unbounded-fan in and $\MOD_{p}$ gates. Additionally, we show a separation between $\QAC^0$, the circuit class containing Toffoli gates with unbounded control, and $\AC^0[p]$, when $\QAC^0$ is augmented with additional mid-circuit measurements and classical fanout. This establishes the first such separation for a shallow-depth quantum class that does not involve quantum fanout gates, while relying solely on finite quantum gate sets. Equivalently, this yields a separation between $\AC^0[p]$ and $[\QNC^0, \AC^0]^2$, i.e., shallow quantum circuits interleaved with simple classical computation.

Secondly, we consider $\QNC^0$ circuits with infinite-size gate sets. We show that these circuits, along with quantum prime modular gates or classical prime modular gates in combination with classical fanout, can implement threshold gates, showing that $\mathsf{QNC}^0[p]=\mathsf{QTC}^0$. Finally, we also show that in the infinite-size gate set case, these quantum circuit classes for higher-dimensional Hilbert spaces do not offer any advantage to standard qubit implementations.
\end{abstract}

\clearpage

\section{Introduction}\label{intro}

In the current landscape dominated by NISQ (Noisy Intermediate-Scale Quantum) devices~\cite{Preskill2018quantumcomputingin}, understanding the computational power of shallow-depth quantum circuits has both practical and theoretical relevance. From a practical perspective, given the noisy aspect of near-term devices, focusing on tasks that achieve quantum advantage in low depth is fundamental. From a theoretical perspective, (classical) constant-depth circuits have been crucial in the development of lower bound techniques, e.g. \cite{rossman15}. Thus, the combination of these techniques, along with the computational capabilities of constant-depth quantum circuits, enables the discovery of new quantum-classical separations and, furthermore, allow a more comprehensive understanding of quantum computational complexity classes \cite{Arora22, jozsa2006, Grier20}.

The holy-grail in quantum circuit complexity is to find problems that can be solved by shallow-depth quantum circuits and that cannot be solved by polynomial-time classical computation. Unfortunately, there are several caveats in the current approaches. Primarily, the current super-polynomial quantum advantage proposals rely on computational assumptions, offering no unconditional proof of quantum superiority \cite{aaronson2011computational,bouland2019complexity}. Additionally, another set of results requires quantum gates with arbitrary accuracy, which cannot be implemented in practice in real-world devices at a constant depth. This problem arises when infinite gate sets are assumed \cite{Briet23,Watts23}, as it has been shown that even if precise quantum operations were feasible, the computational overhead for classical control could negate any quantum advantage \cite{hu2020space}.

More recently, there is a new line of research that has addressed some of these issues, at the cost of achieving weaker statements. More concretely, the seminal work of \cite{Bravyi17} has shown a problem that can be solved by constant-depth quantum circuits ($\QNC^0$), and which cannot be solved by constant-depth classical circuits with bounded fan-in ($\NC^0$). We would like to highlight that while this separation is weaker (since the lower bound is only against constant depth circuits), it is {\em unconditional}, which is rather rare in complexity theory. Moreover, these proposals also have practical interest with their extension to include fault-tolerance~\cite{bravyi2022}.

The result of \cite{Bravyi17} has also been extended to achieve better quantum/classical separations. We now have average-case separations between quantum constant-depth circuits and biased threshold circuits of constant depth \cite{hsieh2024unconditionally,grewal2024improved}, as well as separations against $\AC^0$ circuits with access to $\MOD_2$ gates ($\AC^0[2]$) using constant-depth quantum circuits with quantum advice $\fQNC^0/\mathsf{qpoly}$ \cite{Watts19}. In all these results, it is necessary to show a quantum upper bound and a classical lower bound on circuit size for the same problem. The main challenge lies in developing new techniques for classical lower bounds and quantum upper bounds that allow us to address the following question:

\begin{center}
   {\em What is the largest class of classical circuits that we can separate from constant-depth quantum circuits? }
\end{center}

The main contribution of this work is to prove new separations between constant-depth quantum circuits and constant-depth classical circuits. Additionally, we revisit quantum circuits with infinite gate set and prove new collapses of different quantum circuit classes.

\subsection{Our contributions}\label{contri}

Our first contribution is to show that there exists a set of relational problems, where the inputs $x$ are $n$-bit strings and the outputs $y$ are $m$-bit strings such that the pair $(x,y)$  is in $R\subset \{0,1\}^n\rightarrow \{0,1\}^m$, for which constant-depth quantum circuits on qudits with quantum advice exist, but are not contained in $\AC^0[p]$. 

\begin{result} \label{result:separationQNC}
(Informal; see \Cref{paralsep}) For all primes $p$, there exists a finite gate set on qudits of dimension $p$ such that constant-depth quantum circuits with quantum advice can solve relation problems that are intractable for any polynomial-size $\AC^0[p]$ circuit, i.e.,  $\fQNC^0/\mathsf{qpoly} \nsubseteq \AC^0[p].$
\end{result}

This result extends the $\AC^0[2]$ separation from \cite{Watts19} to general primes.  To prove it,  we generalize the relation problems on modular operations that have been used in \cite{Bravyi17,Watts19}. Roughly, our modular relation problems with parameters $p,q \in \mathbb{N}$ consist of a pair of bit strings $(x,y)$ such that  $|x| \pmod{p} = 0$ iff $|y| \pmod{q} = 0$. Here, $|x|$ and $|y|$ represent the Hamming weights of the respective bit strings. We show that by exploiting higher-dimensional gate sets, we can solve these problems in constant-depth quantum circuits with a copy of a high-dimensional GHZ state. At the same time, we expand the classical lower bound to $\AC^0[p]$. However, this separation achieves only a quantum/classical distinction regarding error probabilities, which differ by a constant value. In order to improve these separations, we then consider the parallel repetition of these modular problems as proposed in ~\cite{Watts19,coudron2021}. In our results, we simplify the proof of \cite{Watts19} by considering the one-sided error of the quantum circuits and using seminal results in circuit complexity such as the Vazirani XOR lemma~\cite{goldreich2011three} and Razborov-Smolensky separations~\cite{Smolensky87}.

We proceed to show that $\QAC^0$ circuits, which allow for mid-circuit measurements and classical but not quantum fanout, can still generate high-dimensional GHZ states. For this, we employ optimal depth and size methods for GHZ state creation\footnote{These methods were originally described for the qubit case and general graph states in \cite{Watts19}, and have also been studied in the context of one-dimensional cluster states in \cite{Peres23,quek2024multivariate}.}, which is based on poor-man's cat states and balanced binary trees as the resource state entanglement structure. This gadget, along with Result~\ref{result:separationQNC}, allows us to show that for every prime $p$, there exists a constant-depth quantum circuit employing Toffoli gates of dimension $p$ that computes a relation unattainable by any polynomial-size $\AC^0[p]$ circuit. Moreover, when Toffoli gates of arbitrary arity can implement any qubit operation, we attain an unconditional separation between this complexity class and all $\AC^0[p]$ classes simultaneously. It is worth noting that this result aligns with the standard definition of $\QAC^0$, as outlined in the existing literature \cite{pade2020,rosenthal2020,fang06,bera11}, while incorporating classical fanout to describe the ability to copy classical information resulting from mid-circuit measurements an arbitrary number of times — a fairly simple operation from a technical perspective. By excluding quantum fanout, these separations are achieved through an alternative multi-qubit quantum operation without relying on the multi-qubit operations utilized in previous results. Furthermore, recent advances and practical demonstrations have suggested that some quantum device platforms inherently support multi-qubit operations \cite{Song24}, and particularly enabling the construction of multi-qubit Toffoli gates \cite{maslov2018use,groenland2020signal,grzesiak2022efficient,bravyi2022constant}.

\begin{result}
(Informal; see \Cref{ssep}) For every prime $p$, there exists a relation problem solvable exactly by $\QAC^0$ circuits with classical fanout, whereas any $\AC^0[p]$ circuit of polynomial size solves an arbitrary random instance of the problem with a negligible probability. Therefore, $\QAC_{\mathsf{cf}}^0 \nsubseteq \AC^0[p]$  holds for all prime $p$. 
\end{result}

This establishes the separation between $\QAC^0$ and $\AC^0[p]$ for any prime $p$ the single addition of classical fanout, the possibly weakest resource so far \cite{green2000,Green02,Browne11}, pushing forward the question of the minimal resource, which was recently asked in \cite{rosenthal2020}.

We notice that in order to achieve such a result with our techniques, it is essential to consider qudits of higher dimensions. Additionally,
we remark that $\QAC_{\mathsf{cf}}^0$ can generate a GHZ state without the need of unbounded parity gates, as proposed in \cite{rosenthal2020}. Our $\QAC_{\mathsf{cf}}^0$ circuits use quantum measurements and adaptive operations to generate the states.
However, the ability to generate this state does not necessarily imply the ability to compute the parity function, since our state-creation protocol is non-unitary, whereas previous mappings between these resources rely on a unitary GHZ state preparation. We conjecture that generating the GHZ state through a unitary procedure (i.e. without measurements) would require quantum fanout. Thus, the question of whether the parity function can be computed in $\QAC^0$ without fanout gates remains open for further investigation \cite{Parham23}.

Note that our $\QAC^0_{\mathsf{cf}}$ circuits can equivalently be realized as $[\QNC^0, \AC^0]^2$ circuits, i.e., $\QNC^0$ circuits interleaved with a single $\AC^0$ classical computation layer. This shows that augmenting shallow-depth quantum circuits with even very limited classical processing can strictly enhance their computational power, yielding the same separation from $\AC^0[p]$.

\medskip 

Next, we provide some limitations on our techniques by showing a classical upper bound on the modular relation problems that we consider.

\begin{result}\label{relationUPm}
(Informal; see \Cref{relationUP}) For every primes $p,q$, the modular relation problem on strings $(x,y)$ such that $|x| \pmod{p} = 0$ iff $|y| \pmod{q} = 0$ can be solved by $\NC^0[q]$ circuits.~\footnote{These circuits are composed of bounded fan-in gates, as in $\NC^0$, but additionally contain a single unbounded fan-in  $\MOD_q$ gate.}
\end{result}

We notice that this result was unknown even for $p = 2$, and it characterizes $\AC^0[p]$ as the largest constant-depth classical circuit class for which a quantum-classical separation can be identified in the plain model using our proposed relation problems. Indeed, it is conjectured that this particular set of modular relation problems might be sufficient to achieve unconditional separations between the aforementioned $\AC^0[p]$ classical circuit classes and less powerful quantum circuit classes, such as $\fQNC^0$. In parallel, separations against classes like $\TC^0$ (or larger classes) have been achieved with different classes of problems but require the interactive setting \cite{Grier20}.

\medskip 

We then switch gears and consider the class of constant-depth quantum circuits with an infinite size gate set and quantum qudit fanout.\footnote{For any prime $p$, qudit fanout for $p$-dimensional qudits is equivalent to $\mathsf{qMOD}_p$ gates, the quantum versions of the classical $\MOD_p$ gates.} In this setting, we show that for any prime $p$,  we have the collapse of the constant-depth quantum hierarchy for quantum circuits with $p$-dimensional qudits, extending the qubit case~\cite{tani16}.

\begin{result}\label{r:fullcollapse}
(Informal; see \Cref{fullcollapse}) For all prime $p$, there exist a constant depth quantum circuits over qudits of dimension $p$ with auxiliary unbounded quantum fanout gates that implement quantum threshold gates, i.e., $\iQNC^0[p]=\iQTC^0$.\footnote{As we describe later, in our notation, the subscript $[p]$ means access to a quantum and $[p]_c$ to a classical unbounded fan-in modular $p$ gate.}
\end{result}

This result establishes computational equivalence across shallow-depth quantum circuits when implemented on systems over prime-dimensional Hilbert spaces. Specifically, our quantum circuits are constructed using Fourier transforms over Abelian groups and this approach allows the creation of circuits for functions over finite fields of the form $\mathbb{F}_p^n \rightarrow \mathbb{F}_p$. These circuits allow the execution of multi-qubit controlled logical operations.

We then demonstrate that constant-depth quantum circuits implemented in higher prime dimensions can be realized using qubits, with only a multiplicative increase in the gate count. Additionally, quantum fanout can be eliminated when using classical modular gates in conjunction with classical fanout.

\begin{result}\label{r:qudits-to-qubits}
(Informal; see \Cref{qubitcol}). For any prime $p$, any constant-depth quantum circuit over qudits of dimension $p$ with threshold gates can be simulated by a constant-depth quantum circuit over qubits with auxiliary $\MOD_q$ gates. Furthermore, for $\iQNC_{\mathsf{cf}}^0$ circuits over qubits, we have $\iQTC^0 \subseteq \iQNC_{\mathsf{cf}}^0[p]_\mathsf{c}$.
\end{result}

This result is achieved by mapping qudit operations onto tensor products of qubits. The qudit operations are produced using exact unitary decomposition methods, as previously demonstrated \cite{reck94,Barenco95}. Specifically, this theorem implies that the qubit-specific hierarchy can be collapsed using any arbitrary classical modular prime gate in combination with classical fanout — a capability previously believed to require either quantum fanout or multi-qubit parity gates. Consequently, we have that $\iQNC_{\mathsf{cf}}^0[p]_\mathsf{c}=\mathsf{QACC}^0$ holds for all primes $p$. This significantly contrasts with classical analog circuit classes, where varying modular gates results in different complexity classes.

From a practical standpoint, this indicates that for any algorithm within these classes, employing qudits instead of qubits offers no significant computational advantage—except potentially for simpler hardware implementations and improved error-correction schemes in Hilbert spaces of varying dimensions. It also demonstrates that quantum circuits over qubits with fixed depth, when augmented with unbounded fan-in classical prime modular operations $\MOD_p$, can implement the quantum subroutines of algorithms such as factoring, which are conjectured to provide exponential quantum advantage \cite{Hoyer05,Browne11}. Furthermore, constant-depth circuits with these gate sets can be realized using standard finite gate sets in logarithmic depth, still resulting in relatively shallow quantum circuits. Additionally, using classical modular gates simplifies implementation, as these gates can be efficiently realized, whereas their quantum counterparts are technically more challenging to implement.

Finally, as future work, we propose extending \Cref{r:qudits-to-qubits} to show that any $\iQTC^0$ circuit can be simulated by constant-depth quantum circuits  over qudits of dimension $p$ with additional classical modular gates of arity $q$, for arbitrary primes $p$ and $q$. This would establish the equivalence between all constant-depth quantum circuits over qudits of different prime dimensions, which could be advantageous for hardware implementations that favor a particular gate sets over a prime-dimensional Hilbert space.

\subsection{Organization} 
Section \ref{preliminaries} introduces all the preliminary definitions necessary for the subsequent results. Section \ref{fqncs} then delves into the study of shallow-depth quantum circuits with finite-size gate sets. In this section, we introduce the concept of modular relation problems and explain how quantum circuits can solve them. Additionally, we establish the $\AC^0[p]$ lower bounds for the circuit classes addressing these problems and conclude the section by defining the $\NC^0[p]$ classical upper bounds for the relation problems under consideration. In Section \ref{quditcolapse}, we shift our focus to constant-depth quantum circuits over infinite-size qudit gate sets and unbounded fan-in modular gates, demonstrating both hierarchical collapses and computational equivalences between qubit and qudit formulations.

\section*{Acknowledgments}

We acknowledge helpful discussions with Robin Kothari, Johannes Frank, Sathyawageeswar Subramanian, Leandro Mendes and Min-Hsiu Hiseh on an early version of these results. MdO is supported by National Funds through the FCT - Fundação para a Ciência e a Tecnologia, I.P. (Portuguese Foundation for Science and Technology) within the project IBEX, with reference PTDC/CCI-COM/4280/2021, and via CEECINST/00062/2018 (EFG). ABG is supported by ANR JCJC TCS-NISQ ANR-22-CE47-0004. ABG, EK and DM are supported by the PEPR integrated project EPiQ ANR-22-PETQ-0007 part of Plan France 2030.

\section{Preliminaries} \label{preliminaries}

In this section, we introduce the relevant classical and quantum circuit classes, along with the basic definitions of qudit quantum computation and helpful auxiliary lemmas and results related to these concepts.

\subsection{Classical circuit classes}

Circuit classes form a crucial aspect of computational theory, helping to categorize and understand the computational capabilities of various models. Certain classical circuit classes emerge as especially relevant when studying constant-depth quantum circuits. This subsection will list and define the main classes of interest while also summarizing pertinent theorems associated with them.

\begin{definition}[$\NC^k$ class] 
The $\mathsf{NC^k}$ class is defined as the computational problems that can be solved in poly-logarithmic depth $log(n)^k$ using polynomially-many gates. The gate set for this class comprises the following bounded fan-in gates,
\begin{figure}[H]
\begin{center}
\scalebox{0.75}{
\begin{tikzpicture}[
    node distance = 1cm,
    operation/.style = {rectangle, draw, minimum size=0.75cm, align=center},
    line/.style = {line width=0.7pt},
]

\node[operation] (and) {AND};
\node[right = 0.2cm of and] (comma1) {\Huge,};
\node[operation, right = 0.2cm of comma1] (or) {OR};
\node[right = 0.2cm of or] (comma2) {\Huge,};
\node[operation, right = 0.2cm of comma2] (not) {NOT};

\node[right = 0.2cm of not] (braceright) {\Huge$\}$};
\node[right = 0cm of braceright] (dot) {\Huge.};

\node[left = 0cm of and] (braceleft) {\Huge$\{$};

\draw[line] ([xshift=-0.1cm, yshift=0.2cm]and.north) -- (and.north);
\draw[line] ([xshift=0.1cm, yshift=0.2cm]and.north) -- (and.north);

\draw[line] ([xshift=-0.1cm, yshift=0.2cm]or.north) -- (or.north);
\draw[line] ([xshift=0.1cm, yshift=0.2cm]or.north) -- (or.north);

\draw[line] ([yshift=0.2cm]not.north) -- (not.north);

\foreach \i in {-0.5,-0.25,...,0.5} {
  \draw[line] ([xshift=\i cm, yshift=-0.15cm]and.south) -- (and.south);
}
\foreach \i in {-0.5,-0.25,...,0.5} {
  \draw[line] ([xshift=\i cm, yshift=-0.15cm]or.south) -- (or.south);
}
\foreach \i in {-0.5,-0.25,...,0.5} {
  \draw[line] ([xshift=\i cm, yshift=-0.15cm]not.south) -- (not.south);
}

\end{tikzpicture}
}
\label{fig:nc0}
\end{center}
\end{figure}
\end{definition}

\begin{definition} [$\AC^k$ class] 
The $\mathsf{AC^k}$ class is defined as the computational problems that can be solved in poly-logarithmic depth $log(n)^k$ using polynomial-many gates. The gate set for this class comprises the following unbounded fan-in gates,
\begin{figure}[H]
\begin{center}
\scalebox{0.75}{
\begin{tikzpicture}[
    node distance = 1cm,
    operation/.style = {rectangle, draw, minimum size=0.75cm, align=center},
    line/.style = {line width=0.7pt},
]

\node[operation] (and) {AND};
\node[right = 0.2cm of and] (comma1) {\Huge,};
\node[operation, right = 0.2cm of comma1] (or) {OR};
\node[right = 0.2cm of or] (comma2) {\Huge,};
\node[operation, right = 0.2cm of comma2] (not) {NOT};

\node[right = 0.2cm of not] (braceright) {\Huge$\}$};
\node[right = 0cm of braceright] (dot) {\Huge.};

\node[left = 0cm of and] (braceleft) {\Huge$\{$};

\foreach \i in {-0.5,-0.25,...,0.5} {
  \draw[line] ([xshift=\i cm, yshift=0.15cm]and.north) -- (and.north);
}

\foreach \i in {-0.5,-0.25,...,0.5} {
  \draw[line] ([xshift=\i cm, yshift=0.15cm]or.north) -- (or.north);
}

\draw[line] ([yshift=0.15cm]not.north) -- (not.north);

\foreach \i in {-0.5,-0.25,...,0.5} {
  \draw[line] ([xshift=\i cm, yshift=-0.15cm]and.south) -- (and.south);
}
\foreach \i in {-0.5,-0.25,...,0.5} {
  \draw[line] ([xshift=\i cm, yshift=-0.15cm]or.south) -- (or.south);
}
\foreach \i in {-0.5,-0.25,...,0.5} {
  \draw[line] ([xshift=\i cm, yshift=-0.15cm]not.south) -- (not.south);
}

\end{tikzpicture}
}
\label{fig:ac0}
\end{center}
\end{figure}
\end{definition}

\begin{definition}$(\AC^k[p]\ \mathsf{class})$. 
The $\AC^k[p]$ class is defined as the computational problems that can be solved in poly-logarithmic depth $log(n)^k$ using polynomial-many gates. The gate set for this class comprises the following unbounded fan-in gates,
\begin{figure}[H]
\begin{center}
\scalebox{0.75}{
\begin{tikzpicture}[
    node distance = 1cm,
    operation/.style = {rectangle, draw, minimum size=0.75cm, align=center},
    line/.style = {line width=0.7pt},
]

\node[operation] (and) {AND};
\node[right = 0.2cm of and] (comma1) {\Huge,};
\node[operation, right = 0.2cm of comma1] (or) {OR};
\node[right = 0.2cm of or] (comma2) {\Huge,};
\node[operation, right = 0.2cm of comma2] (mod) {$\MOD_p$};
\node[right = 0.2cm of mod] (comma3) {\Huge,};
\node[operation, right = 0.2cm of comma3] (not) {NOT};
\node[right = 0.2cm of not] (braceright) {\Huge$\}$};

\node[right = 0cm of braceright] (comma) {\Huge,};

\node[left = 0cm of and] (braceleft) {\Huge$\{$};

\foreach \i in {-0.5,-0.25,...,0.5} {
  \draw[line] ([xshift=\i cm, yshift=0.15cm]and.north) -- (and.north);
}

\foreach \i in {-0.5,-0.25,...,0.5} {
  \draw[line] ([xshift=\i cm, yshift=0.15cm]or.north) -- (or.north);
}

\foreach \i in {-0.5,-0.25,...,0.5} {
  \draw[line] ([xshift=\i cm, yshift=0.15cm]mod.north) -- (mod.north);
}

\draw[line] ([yshift=0.15cm]not.north) -- (not.north);

\foreach \i in {-0.5,-0.25,...,0.5} {
  \draw[line] ([xshift=\i cm, yshift=-0.15cm]and.south) -- (and.south);
}
\foreach \i in {-0.5,-0.25,...,0.5} {
  \draw[line] ([xshift=\i cm, yshift=-0.15cm]or.south) -- (or.south);
}
\foreach \i in {-0.5,-0.25,...,0.5} {
  \draw[line] ([xshift=\i cm, yshift=-0.15cm]mod.south) -- (mod.south);
}
\foreach \i in {-0.5,-0.25,...,0.5} {
  \draw[line] ([xshift=\i cm, yshift=-0.15cm]not.south) -- (not.south);
}
\end{tikzpicture}
}
\label{fig:ac0p}
\end{center}
\end{figure}
\noindent with $\MOD_p$ gate defined as follows,
\begin{equation}
\mathsf{MOD_p}(x) = 
\begin{cases} 
1 & \text{if } \sum_{i=1}^{n} x_i \mod p=0 \\
0 & \text{if } \sum_{i=1}^{n} x_i \mod p\neq 0 
\end{cases}. 
\end{equation} 
\end{definition}

The Razborov-Smolensky separations between all the $\AC^0[p]$ classes will be fundamental to establishing new separations between these and quantum constant depth circuit classes.

\begin{theorem}\label{Acseparations} \cite{Smolensky87}
Let $p$ be a prime number, and $q$ is not a power of $p$, then computing $\MOD_p$ with a depth $d$,  $\AC^0[q]$ circuit, requires $\exp(\Omega(n^{1/2d}))$ size.
\end{theorem}
 
\noindent Furthermore, it should be noted that this theorem implies that when we construct a graph using the $\AC^0[q]$ classes as vertices and draw edges between pairs of classes that are not subsets of each other, we obtain a complete graph. Thus, for each edge, at least two problems distinguish one class from the other (see Figure \ref{fig:containements}).

\begin{figure}[H]
\begin{center}
\begin{tikzpicture}
    \def\radius{3cm}
    
    \coordinate (Apos) at ({\radius*cos(0)}, {\radius*sin(0)});
    \coordinate (Bpos) at ({\radius*cos(45)}, {\radius*sin(45)});
    \coordinate (Cpos) at ({\radius*cos(90)}, {\radius*sin(90)});
    \coordinate (Qpos) at ({\radius*cos(135)}, {\radius*sin(135)});
    \coordinate (Epos) at (0, -1.3);  
    
    \foreach \position/\name in {Apos/$\AC^0[2]$, Bpos/$\AC^0[3]$, Cpos/$\AC^0[5]$, Epos/$\AC^0[p]$} {
        \fill (\position) circle (2pt);
        \node[right=0.25cm of \position, font=\small] {\name};
    }
    
    \fill (Qpos) circle (2pt);
    \node[left=0.25cm of Qpos, font=\small] {$\AC^0[q]$};
    
    \foreach \from/\to/\label in {
        Apos/Bpos/{$\MOD_2,\MOD_3$},
        Cpos/Qpos/",
        Apos/Cpos/",
        Apos/Epos/{$\MOD_2,\MOD_p$},
        Bpos/Epos/",
        Cpos/Epos/"
    } {
        \draw (\from) -- (\to) node[midway, above, font=\tiny, sloped] {\label};
    } 

    \draw (Apos) -- (Qpos) node[pos=0.47, above, font=\tiny, sloped] {$\MOD_2,\MOD_q$};
    \draw (Bpos) -- (Qpos) node[pos=0.6, above, font=\tiny, sloped] {"};
    \draw (Bpos) -- (Cpos) node[pos=0.3, above, font=\tiny, sloped] {"};
    \draw (Qpos) -- (Epos) node[pos=0.5, above, font=\tiny, sloped] {$\MOD_p,\MOD_q$};
    \draw (Qpos) -- (Epos) node[pos=0.5, below, font=\tiny, sloped] {$\nsupseteqq\ ,\ \nsubseteq$};
    
\end{tikzpicture}
\caption{Graph representation of 
$\AC^0[p]$ classes with $p$ prime. Each edge denotes that the connected classes are distinct, without one being a subset of the other. Along the edges, examples highlight problems that are in one class but excluded from the other.}
\label{fig:containements}
\end{center}
\end{figure}

\begin{definition}[$\TC^k$ class] The $\TC^k$ class is defined as the computational problems that can be solved in poly-logarithmic depth $log(n)^k$ using polynomial-many gates. The gate set for this class comprises the following unbounded fan-in gates,
\begin{figure}[H]
\begin{center}
\scalebox{0.75}{
\begin{tikzpicture}[
    node distance = 1cm,
    operation/.style = {rectangle, draw, minimum size=0.75cm, align=center},
    line/.style = {line width=0.7pt},
]

\node[operation] (and) {AND};
\node[right = 0.2cm of and] (comma1) {\Huge,};
\node[operation, right = 0.2cm of comma1] (or) {OR};
\node[right = 0.2cm of or] (comma2) {\Huge,};
\node[operation, right = 0.2cm of comma2] (mod) {$\mathsf{TH}_t$};
\node[right = 0.2cm of mod] (comma3) {\Huge,};
\node[operation, right = 0.2cm of comma3] (not) {NOT};
\node[right = 0.2cm of not] (braceright) {\Huge$\}$};

\node[right = 0cm of braceright] (comma) {\Huge,};

\node[left = 0cm of and] (braceleft) {\Huge$\{$};

\foreach \i in {-0.5,-0.25,...,0.5} {
  \draw[line] ([xshift=\i cm, yshift=0.15cm]and.north) -- (and.north);
}

\foreach \i in {-0.5,-0.25,...,0.5} {
  \draw[line] ([xshift=\i cm, yshift=0.15cm]or.north) -- (or.north);
}

\foreach \i in {-0.5,-0.25,...,0.5} {
  \draw[line] ([xshift=\i cm, yshift=0.15cm]mod.north) -- (mod.north);
}

\draw[line] ([yshift=0.15cm]not.north) -- (not.north);

\foreach \i in {-0.5,-0.25,...,0.5} {
  \draw[line] ([xshift=\i cm, yshift=-0.15cm]and.south) -- (and.south);
}
\foreach \i in {-0.5,-0.25,...,0.5} {
  \draw[line] ([xshift=\i cm, yshift=-0.15cm]or.south) -- (or.south);
}
\foreach \i in {-0.5,-0.25,...,0.5} {
  \draw[line] ([xshift=\i cm, yshift=-0.15cm]mod.south) -- (mod.south);
}
\foreach \i in {-0.5,-0.25,...,0.5} {
  \draw[line] ([xshift=\i cm, yshift=-0.15cm]not.south) -- (not.south);
}

\end{tikzpicture}
}
\label{fig:tc0}
\end{center}
\end{figure}
\noindent with $\mathsf{TH}_k$ gate defined as follows,
\begin{equation}
\mathsf{TH_t}(x) = 
\begin{cases} 
1 & \text{if } \sum_{i=1}^{n} x_i \geq t \\
0 & \text{if } \sum_{i=1}^{n} x_i < t
\end{cases}. 
\end{equation} 
\end{definition}

We introduce a class that is not standard in circuit complexity literature but holds significance for the discussions and results presented in this text.

\begin{definition} $(\NC^0[p]$ class$)$.
The $\NC^k[p]$ class is defined as the computational problems that can be solved in poly-logarithmic depth $log(n)^k$ using polynomial-many gates. The gate set for this class comprises the bounded fan-in $\mathsf{AND}$, $\mathsf{OR}$, $\mathsf{NOT}$ gates, and a single unbounded fan-in $\mathsf{Mod}_p$ gates. 
\end{definition}

A considerable number of important results in computational complexity pertain to the circuit classes we have defined, as well as their interrelations \cite{furst1984parity,Smolensky87,rossman15,arora2009computational}. Consequently, we will present the established containments for these classes,
\begin{equation}
    \NC^0 \subsetneq \AC^0,\NC^0[p] \subsetneq \AC^0[p] \subsetneq \TC^0 \subseteq \NC^1.
\end{equation}

\noindent This knowledge is crucial for comprehending the capabilities of each circuit class and discerning the position of quantum circuit classes relative to them.

Lastly, we will present the Vazirani-XOR lemma, which is of great interest to the study of parallel repeated games. This lemma is of particular interest for the analysis of constant-depth circuits and will be used to improve our classical hardness results.

\begin{lemma}\label{Xor}
(Vazirani’s XOR Lemma \cite{goldreich2011three}). Let $D$ be a distribution on $\mathbb{F}_2^m$ and $\chi_S$ denote the parity function on the set $S \subseteq [m]$, defined as $\chi_S(x) = \oplus_{i\in S} x_i$. If $|\mathbb{E}_{x\in D}[(-1)^{\chi_S(x)}]| \leq \epsilon$ for every nonempty subset $S \subseteq [m]$, then $D$ is $\epsilon \cdot 2^{m/2}$-close in statistical distance to the uniform distribution over $\mathbb{F}_2^m$.
\end{lemma}

\subsection{Quantum circuit classes}

The initial quantum circuit class under examination describes the quantum analog of $\NC^0$ and comprises a gate set defined by a collection of qudit gates that have bounded fan-in.  This includes, for instance, standard universal gate sets like the Clifford+T gate set.

\begin{definition} [$\fQNC^0$]
For $d \in \mathbb{N}$, let $\mathcal{H}^d$ denote a $d$-dimensional Hilbert space.  We define the $\QNC^i$ class as the set of quantum circuits composed of bounded fan-in quantum gates acting on qudits in $\mathcal{H}^d$, with depth $O(\log^i (n))$, and a polynomial number of gates concerning the size of the input string $n$.
\end{definition}

Additionally, we consider access to quantum advice, which is a quantum state that depends on the input size but not on the input itself. In particular, we examine $\fQNC^0/\mathsf{qpoly}$, which consists of $\fQNC^0$ circuits augmented with polynomial-size quantum advice. We also explore extensions of certain constant-depth circuit classes, including unbounded fan-in modular operations. For instance, the class $\iQNC^0[p]$ combines $\iQNC^0$ circuits with an additional quantum modular gate defined as follows,   
\begin{equation}
\mathsf{qMOD}_p \ket{x_0, x_1, \hdots, x_n} := \ket{x_0+(x_1 +\hdots +x_n) \mod p,x_1,\hdots,x_n\mod p}.
\end{equation}
\noindent Finally, we will use the suffix $[p]_c$ to denote the inclusion of a classical unbounded fan-in modular gate $\MOD_p$ in these circuit classes.

The next class to be analyzed will expand on the $\iQNC^0$ circuit class by incorporating different multi-qudit gates with the intent of introducing the characterization of the quantum counterpart of $\AC^0$.

\begin{definition}[$\fQAC^0$]
For $d \in \mathbb{N}$, let $\mathcal{H}^d$ denote a $d$-dimensional Hilbert space.  We define the $\QAC^i$ class as the set of quantum circuits composed of bounded fan-in quantum gates acting on qudits in $\mathcal{H}^d$ , along with qudit Toffoli gates that have a polynomial number of control qudits. These circuits have depth $O(\log^i (n))$, and a polynomial number of gates concerning the size of the input string $n$.
\end{definition}

In addition, we define $\QAC_{\mathsf{cf}}^0$ as the circuit class describing $\QAC^0$ circuits with auxiliary
mid-circuit measurements and classical fanout. This additional classical fanout not only ensures a fair comparison between the quantum and classical constant-depth circuit classes—since the latter have access to this resource—but also allows us to study the additional capability of copying classical information resulting from mid-circuit measurements an arbitrary number of times, which is a natural computational operation. Furthermore, this does not conflict with the no-cloning theorem, as it operates solely on classical data, and it is a straightforward operation to implement in hardware, in contrast with its quantum counterpart.

\subsection{Quantum computation with qudits}\label{operations}

In this preliminary section, we present the qudit gates that will be employed in subsequent sections of the text to construct the quantum circuits essential for our main proofs.

The first qudit operation that we will define will be a generalized version of the qubit $X$ gate.

\begin{definition}[X qudit gate]
The $\mathsf{X}$ gate for qudits is defined as follows, 
\begin{equation}
    \mathsf{X}_d  \ket{m}:= \ket{m+1 \mod d} \ .
\end{equation}
\end{definition}

\noindent Then a rotation about the Z axis will be defined as follows,
\begin{definition}[Rotation Z qudit gate]
The rotation Z operator for qudits is defined as follows, 
\begin{equation}
    \mathsf{R}_Z^d(\phi) := \sum_{j=0}^{d-1} e^{i(1-sgn(d-1-j))\phi} \ket{j}\bra{j} \ .
\end{equation}
\end{definition}

\noindent In addition, we will define and make use of a more complex rotation Z operator.

\begin{definition}[Generalized rotation Z qudit gate]
The generalized rotation Z operator for qudits is defined as follows, 
\begin{equation}
\mathsf{GR}_Z^d(\phi) := \sum_{j=0}^{d-1} e^{i(\phi j)} \ket{j}\bra{j} \ .
\end{equation}
\end{definition}

\noindent The last single qudit operation will be the qudit Hadamard gate.

\begin{definition}[Fourier gate]
The Fourier gate is defined as follows
\begin{equation}
    \mathsf{F}_d \ket{m} = \frac{1}{\sqrt{d}} \sum_{n \in \mathbb{Z}_d} \omega^{mn} \ket{n},
\end{equation}
\noindent with $\omega=e^{i\frac{2\pi}{d}}$.
\end{definition}

Now the respective two qudit gates of interest are the qudit equivalent of the 
$\mathsf{CNOT}$ gate and the controlled rotation Z gate. 

\begin{definition}[$\mathsf{SUM}$ gate]
 The $\mathsf{SUM}$ gate is defined as follows,
\begin{equation}
    \mathsf{SUM} \ket{n}\ket{m}  = \ket{n}\ket{n+m} .
\end{equation}
\end{definition}

\begin{definition}[Controlled generalized rotation Z qudit gate]
The $\mathsf{CGR}_Z^d$ gate is defined as follows,
\begin{equation}
\mathsf{CGR}_Z^d \ket{n}\ket{m}  = \ket{n} (\mathsf{GR}_Z^d)^n \ket{m}.
\end{equation}
\end{definition}

\noindent The last two multi-qudit operations are the following. 

\begin{definition}[$\mathsf{fanout}_p$]
The $\mathsf{fanout}_p$ gate is defined as follows,
\begin{equation}
\mathsf{fanout}_p \ket{x_0, x_1, \hdots, x_n} := \ket{x_0,(x_0+x_1)\mod p,\hdots,(x_0+x_n)\mod p},
\end{equation}
\noindent and its inverse $\mathsf{fanout}_p^{-1}$ can be implemented with $(\mathsf{fanout}_p)^{p-1}$.
\end{definition}

\begin{definition}[$\mathsf{qMOD}_p^{-1}$]
The $\mathsf{qMOD}_p^{-1}$ gate is defined as follows,
\begin{equation}
\mathsf{qMOD}_p^{-1} \ket{x_0, x_1, \hdots, x_n} := \ket{x_0+p-(x_1 +\hdots +x_n) \mod p,x_1,\hdots,x_n},
\end{equation}
\noindent and can be defined and implemented using $(\mathsf{qMOD}_p)^{p-1}$.
\end{definition}

Finally, we also define two key orthogonal bases for qudit states that will be used in our works. The first one is an extension of the qubit X-basis for qudits.
\begin{definition} [Qudit orthogonal X-basis]
The qudit orthogonal X-basis can be described by the set of states, 
\begin{equation}
\ket{X_d^m}=\frac{1}{\sqrt{d}} \sum_{j=0}^{d-1} \omega^{j\cdot m}\ket{j}\ 
\end{equation}
\noindent with $\omega=e^{\frac{i 2\pi}{d} }$ and $m \in \{0,1,..., d-1\}$.
\end{definition}

\noindent The second orthogonal basis is tailored specifically for qudit GHZ states and will be necessary for the analysis of the quantum circuits we will present. 

\begin{definition}[Qudit-GHZ orthogonal X-basis] 
The qudit-GHZ orthogonal X-basis can be defined by the set of states, 
\begin{equation}
\ket{\mathsf{GHZ}_{d,n}^m}=\frac{1}{\sqrt{d}} \sum_{j=0}^{d-1} \omega^{j\cdot m}\ket{j}^{\otimes n}\ 
\end{equation}
\noindent with $\omega=e^{\frac{i 2\pi}{d} }$ and $m \in \{0,1,...,d-1\}$.
\end{definition}

Finally, we state two technical lemmas regarding these bases that are needed in our results.



\begin{lemma}\label{fullbasis}
$
\braket{X_d^m|X_d^n}=\braket{\mathsf{GHZ}_d^m|\mathsf{GHZ}_d^n}= \delta_{m,n}$. 
\end{lemma}
\begin{proof}
All the elements of the qudit orthogonal X-basis can be represented as, 
\begin{equation}
\ket{X_d^m}=\frac{1}{\sqrt{d}} \sum_{i=0}^{d-1} \omega^{i\cdot m}\ket{i}\ = \mathsf{F}_d \ket{m}\ .
\end{equation}
This yields that,  
\begin{equation}
\braket{X_d^m|X_d^n} = \bra{m} \mathsf{F}_d \mathsf{F}_d^{\dagger} \ket{n}.
\end{equation}
\noindent Combining this with the fact that the $\mathsf{F}_d^{\dagger} \mathsf{F}_d\ket{m}=\ket{m} $, we obtain that $\braket{X_d^m|X_d^n}= \delta_{m,n}$. The same result holds for the qudit-GHZ orthogonal X-basis, given that the bases are simply a tensor product over $n$ elements of the same bases as in the qudit X-basis.
\end{proof}

\begin{lemma}\label{projection}
$\mathsf{F}_d^{\otimes n}\ket{\mathsf{GHZ}_{d,n}^m} = \frac{1}{\sqrt{d^{n-1}}} \sum\limits_{\substack{x \in \mathbb{F}_d^n \\  |x| \mod d= -m }} \ket{x}$.
\end{lemma} 

\begin{proof}
By considering the effect of the operator on the basis states we determine that,
\begin{align}
\mathsf{F}_d^{\otimes n}\ket{\mathsf{GHZ}_d^m}&=\frac{1}{\sqrt{d}} \sum_{i=0}^{d-1} \big( \omega^{i\cdot m} \mathsf{F}_d^{\otimes n}\ket{i}^{\otimes n}\big )\\
&=\frac{1}{\sqrt{d}} \sum_{i=0}^{d-1} \omega^{i\cdot m} \big( \frac{1}{\sqrt{d^{n}}}\sum_{x \in \{0,1,...,d-1\}^n} \omega^{i\cdot |x|} \ket{x} \big ) \\
&=\frac{1}{\sqrt{d^{n+1}}} \sum_{x \in \{0,1,...,d-1\}^n} \sum_{i=0}^{d-1} \omega^{i\cdot (m+|x|)}  \ket{x} \ 
\end{align}

\noindent with $|x|=\sum_{i=0}^{n-1} x_i \mod d$. This implies that for all the inputs $x$ for which $m+|x| \mod d=0$ the value of $\omega^{i\cdot (m+|x|)}=1$. In contrast, for inputs $x$ for which $m+|x| \mod d \neq 0$ we will use the geometric progression to demonstrate first that the sum of all roots of unity is equal to zero, 
\begin{equation}
    \sum_{i=0}^{d-1} \omega^i = \frac{1-\omega^d}{1-\omega}=0\ .
\end{equation}
\noindent Then, with one more step, we can show that the same is true for the amplitudes of the inputs considered, 
\begin{align}
\sum_{i=0}^{d-1} \omega^{i\cdot (m+|x|)} =  \frac{1-\omega^{d\cdot (m+|x|)}}{1-\omega^{m+|x|}} = 0 \ .
\end{align}

Finally, this implies that only for values that $m+|x| \mod d=0$ the amplitude of the respective basis states $\ket{x}$ are non-zero. This produces exactly the states, \begin{equation}
\frac{1}{\sqrt{d^{n-1}}} \sum_{x \in X} \ket{x},\ with\ X=\big \{ y\  |\ y\in \{0,1,...,d-1 \}^n ,\  |y| \mod d= -m \big\} \ .
\end{equation}
\end{proof}

\section{Shallow depth quantum circuits with finite-size gate sets}\label{fqncs}

In this section, we show new separations between quantum shallow circuits and classical ones. Moreover, in \Cref{sec:f-upper-bounds}, we prove the quantum upper bounds, and in \Cref{sec:f-lower-bounds}, we prove the classical lower bounds. We combine these results in \Cref{sec:f-separations} in order to achieve our new separations. Finally, in \Cref{sec:f-limitations}, we discuss the limitations of our techniques.

\subsection{Quantum upper bounds}\label{sec:f-upper-bounds}

In this subsection, we introduce lower bounds for the probability that $\fQNC^0$ circuits with quantum advice and $\QAC_{\mathsf{cf}}^0$ circuits can solve a type of relation problem, which we will define as modular relation problems and subsequently use for the intended separations.

\begin{definition}[Modular relation problem]\label{defmod}
The modular relation problem $\mathcal{R}_{q,p}^m:\{0,1\}^n \rightarrow \{0,1\}^{m} $ is defined as
\begin{equation}
\mathcal{R}_{q,p}^m (x) = \left \{ y\ \big |\ y \in \mathbb{F}_2^{m} ,\ |y| \pmod{q} = 0 \text{ iff } 
|x| \pmod{p} = 0
\right \}.
\end{equation} 
\end{definition}

This class of problems subsumes the relation problem used in \cite{Bravyi17,Watts19}, which uses $\mathcal{R}_{2,4}^{o(n^2)}$ for the $\NC^0$ and the $\AC^0$ separations, and \cite{Watts19}, that uses $\mathcal{R}_{2,3}^{o(n^2)}$ for their $\AC^0[2]$ separation.

\subsubsection{Shallow depth quantum circuits with advice}
\label{quantcirc}

In this subsection, we will introduce quantum circuits with a quantum advice that compute $\mathcal{R}_{q,p}^{q\cdot n}$ with a one-sided error. 

\begin{theorem}\label{quantumcirc}
For all fixed and distinct primes $p$ and $q$, 
$\mathcal{R}_{q,p}^{q\cdot n}$ can be solved in $\fQNC^0/\mathsf{qpoly}$ by a circuit over qudits of dimension $q$, with one-sided error at most $\frac{1-\cos(\frac{2\pi}{p})}{q}$.
\end{theorem}
\begin{proof}
The circuit for $\mathcal{R}_{q,p}^{q\cdot n}$ problem is described in Figure \ref{fig:quditcircuits}.

The initial state is $\ket{x}\ket{\mathsf{GHZ}_{q,n}^0}$, where the GHZ state is the quantum advice of the circuit. Then, we apply a layer of $n$ control-$Z$ rotations on the $n$ qubits with a phase equal to $\frac{2\pi}{p}$. The  $i$-th bit of the input $x$ acts as the control qubit to the control-$Z$ with the target being the $i$-th qubit of the GHZ state.

The first register is then
\begin{equation}
 \ket{\psi^{(1)}}= \frac{1}{\sqrt{q}} \Big( \sum_{i=0}^{q-2} \ket{i}^{\otimes n} + e^{i\frac{2\pi|x|}{p}  }\ket{q-1}^{\otimes n} \Big) 
 = \sum_{i=0}^{q-1} c_i \ket{\mathsf{GHZ}_{q,n}^i},
\label{eq:decomposition-ghz}
\end{equation}
\noindent with the second equality being the decomposition of $\ket{\psi^{(1)}}$ in the qudit-GHZ orthogonal X-basis. 

If we apply $F^{\otimes n}_q$ to $\ket{\psi^{(1)}}$, we have
\begin{align*}
\ket{\psi^{(2)}} = \mathsf{F}_q^{\otimes n}\ket{\psi^{(1)}} &=  \sum_{i=0}^{q-1} c_i\ (\mathsf{F}_q^{\otimes n}\ket{\mathsf{GHZ}_{q,n}^i})=
\sum_{i=0}^{q-1} c_i
\sum_{\substack{z \in [q]^n \\ |z| \mod q= -i}} \ket{z},
\end{align*}
\noindent where the last equality comes from \Cref{projection}.

 Our goal is to compute the probability of measuring this state and get an output $y$ such that $|y| \pmod{p} = 0$. For that, notice that $c_0$  is:
\begin{align}
 \braket{\mathsf{GHZ}_{q,n}^0|\psi^{(1)}} &=  \frac{1}{\sqrt{q}} \sum_{i=0}^{q-1} \bra{i}^{\otimes n} \frac{1}{\sqrt{q}} \Big( \sum_{i=0}^{q-2} \ket{i}^{\otimes n} + e^{i\frac{2\pi|x|}{p}}\ket{q-1}^{\otimes n} \Big) = \frac{q-1+\cos(\frac{2\pi |x|}{p})}{q}\ . \label{eq:value_c0}
\end{align}

We then divide that analysis into two cases:
\begin{enumerate}
    \item  If $|x|\mod p = 0$, $(c_0)^2$ equals 1. Consequently, measuring $\ket{\psi_2}$ leads to a string $y$ such that $|y|\mod q=0$. 
    \item If $|x|\mod p\neq 0$, we have $c_0^2 \leq \frac{q-1+\cos(\frac{2\pi |x|}{p})}{q}$. This implies that with a probability of least $\frac{1-\cos(\frac{2\pi}{p})}{q}$, 
it holds that  $|y| \mod q \neq 0$. 
\end{enumerate}

The remaining step consists of mapping $y \in [q]^n$ to bit-strings $y' \in \{0,1\}^{q\cdot n}$ such that the Hamming weight is preserved. This can be done with extra ancilla qubits by mapping each $y_i \in [q]$ to the bit-string $1^{y_i}0^{q-y_i}$. We denote the circuit that performs such an operation as $C(q,2)$.
\end{proof}

\begin{figure}
    \centering
    \includegraphics[scale=0.38]{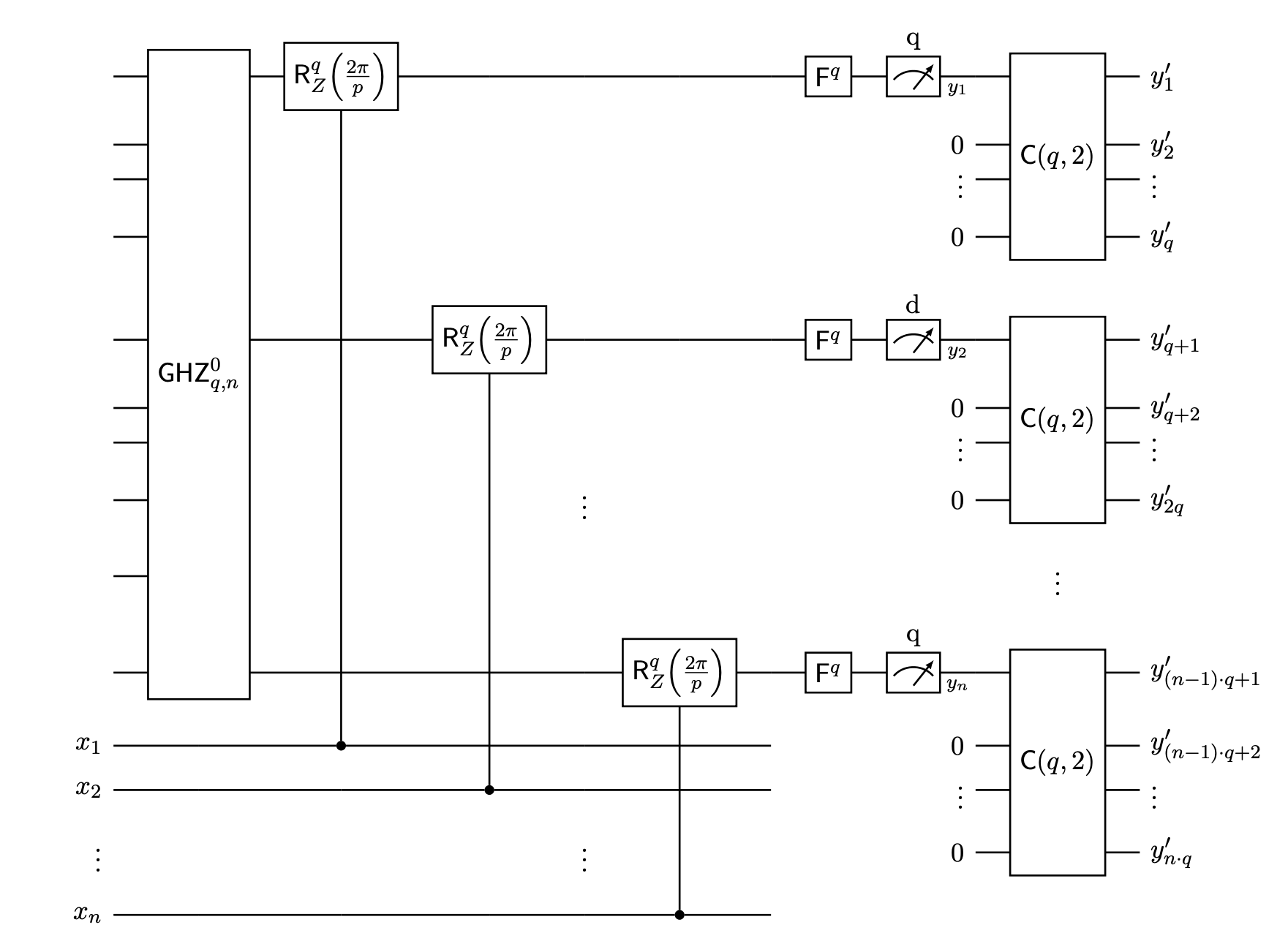}
    \caption{Parameterized quantum circuit class for values $q$, $p$, and $n$, incorporating the advice quantum state $\ket{\mathsf{GHZ}_{q,n}^0}$, which solves the $\mathcal{R}_{q,p}^{q\cdot n}$ with bounded one-sided error.}
    \label{fig:quditcircuits}
\end{figure}

\subsubsection{Shallow depth quantum circuits with Toffoli gates}\label{fqac} 

We show now that quantum circuits of constant depth built with multi-qudit controlled Toffoli gates along with intermediate measurements and unbounded classical fanout can generate qudit-GHZ states, demonstrating that $\fQAC_{\mathsf{cf}}^0$ (without advice) can solve $\mathcal{R}_{q,p}^{q\cdot n}$ with a one-sided error.

To achieve these qudit GHZ states, we use a specialized entanglement structure built upon poor-man's cat states, as referenced in \cite{Watts19,Watts23,Liu2022depthefficient}. Specifically, poor-man's cat states are formulated using graphs where the number of vertices is equal to the state's qubit count, and the edge configuration guarantees a unique path between all vertices.  In our work, we fix this graph to balanced binary trees $\mathcal{B}_n$ and extend these quantum states to higher dimensions.

\begin{definition}[Binary tree-structured poor-man's qudit state).]\label{poorghz}
Let $q$ be a prime. We consider the balanced binary tree $\mathcal{B}_n=(V,E)$, where $V = \{1,...,n\}$ and $E = \{\{i, \left\lfloor \frac{i}{2} \right\rfloor\} : i \in \{2,...,n\}\}$.  As an example, we depict $\mathcal{B}_6$ in \Cref{fig:def-binarytree}.

We define a binary tree-structured poor-mans qudit state $\ket{\mathsf{BPM}_{q,n}}$ as follows
\begin{equation}
\ket{\mathsf{BPM}_{q,n}}=\frac{1}{\sqrt{q^{n}}} \sum_{e_2,...,e_{n},v_1 \in \mathbb{F}_q} \ket{e_2,...,e_{n}}\ket{v_1,v_2,v_3,\hdots,v_n},
\end{equation} 
where each $v_i$, $i \in \{2,...,n\}$  is recursively defined as
\begin{align}
v_i = e_i - v_{\lfloor \frac{i}{2} \rfloor}.
\end{align}

\end{definition}

\begin{figure}[h]
\begin{center}
\begin{tikzpicture}[level distance=1.33cm,
  level 1/.style={sibling distance=6cm},
  level 2/.style={sibling distance=3cm}]
\node (root) {$1$}
  child {node (l) {$2$}
    child {node (ll) {$4$} 
      edge from parent node[left] {$e_{3}$}
    }
    child {node (lr) {$5$} 
      edge from parent node[right] {$e_{4}$}
    }
    edge from parent node[left] {$e_{1}$}
  }
  child {node (r) {$3$}
    child {node (rl) {$6$}
      edge from parent node[left] {$e_{5}$}
    }
    edge from parent node[right] {$e_{2}$}
  };
\end{tikzpicture}
\end{center}
\caption{Representation of $\mathcal{B}_6$.}
\label{fig:def-binarytree}
\end{figure}

We now show that these states can be created using constant-depth quantum circuits.

\begin{lemma}\label{poorconst}
For $n \in \mathbb{N}$, $\ket{\mathsf{BPM}_{q,n}}$ can be prepared in $\fQNC^0$ with a quantum circuit over qudits of dimension $q$.
\end{lemma}
\begin{proof}
For creating this state, we start with a single qudit for each vertex and edge in $\mathcal{B}_n$, i.e., we have $|V|+|E|$ qudits, where the $i$-th qudit, for $1 \leq i \leq n-1$ corresponds to the edge  $\{i+1,\lfloor \frac{i+1}{2} \rfloor \}$ and the $i$-th qudit, for $n \leq i \leq 2n-1$ corresponds to the vertex $i-n+1$. We label qudits by their corresponding  vertices/edges.

We start by applying the Fourier gate $\mathsf{F}_q$ to each of the vertex qudits and then, for edge $e = \{u,v\} \in E$ we apply a $\mathsf{SUM}_{u,e}$ and  $\mathsf{SUM}_{v,e}$.  The overall state after these operations is:
\begin{align}
&\sum_{v_1,...,v_n \in \mathbb{F}_q} \frac{1}{\sqrt{q^{n}}} \ket{(v_1+v_2),(v_1 + v_3)..., (v_n + v_{\frac{n}{2}})}\otimes \ket{v_1,v_2,v_3,\hdots,v_n} \\
&= \sum_{e_1,v_1,v_3...,v_n \in \mathbb{F}_q} \frac{1}{\sqrt{q^{n}}} \ket{e_1,(v_1 + v_3)..., (v_n + v_{\frac{n}{2}})}\otimes \ket{v_1,e_1-v_1,v_3,\hdots,v_n} \\
&= \sum_{e_1,e_2,v_1,v_4...,v_n \in \mathbb{F}_q} \frac{1}{\sqrt{q^{n}}} \ket{e_1,e_2..., (v_n + v_{\frac{n}{2}})}\otimes \ket{v_1,e_1-v_1,e_2 - v_1, v_4,\hdots,v_n} \\
&... \\
&\sum_{e_1,...,e_{n-1},v_1 \in \mathbb{F}_q} \frac{1}{\sqrt{q^{n}}} \ket{e_1,...,e_{n-1}}\otimes \ket{v_1,e_1-v_1,e_2 - v_1, \hdots,e_{n-1}-v_{\frac{n}{2}}}\\
&= \ket{\mathsf{BPM}_{q,n}},
\end{align}
where all operations are $\pmod{q}$ and all the equalities hold by rewriting the elements in $\mathbb{F}_q$.

Finally, given that the $\mathsf{F}_q$ gates are all applied in parallel to the vertex qudits, and the chromatic number of a balanced binary tree plus one is equal to 3, all the $\mathsf{SUM}_{i,j}$ gates can be parallelized in order to be implemented in constant depth.
\end{proof}

We now show that these states can be easily converted to GHZ states.

\begin{lemma}\label{reduc}
There is a $\fQAC_{\mathsf{cf}}^0$ circuit $C_n$ over qudits of dimension $q$  that maps $\ket{\mathsf{BPM}_{q,n}}$ to the n-qudit state $\ket{\mathsf{GHZ}_{q,n}^0}$.
\end{lemma}

\begin{proof}
In the circuit $C_n$, we first measure the edge qudits leading to the output string $e_1...e_{n-1} \in \mathbb{F}^{n-1}_q$. The collapsed quantum state is 

\begin{equation}
\ket{e_1,...,e_{n-1}} \otimes \frac{1}{\sqrt{q}}\sum_{v_1 \in \mathbb{F}_q}  \ket{v_1,v_2,..., v_n}.
\end{equation} 
where $v_i = e_i - v_{\lfloor \frac{i}{2} \rfloor}$. Notice that by opening this recurrence relation, we have that

\begin{align}
    v_i = \sum_{j = 0}^{depth(i)-1} (-1)^{depth(i) - depth(q^{j}(i))} e_{q^{j}(i)} + (-1)^{depth(i)+1} v_1,
\end{align}
where $depth(i)$ denotes the depth of node $i$ in the tree and $q^j(i)$ denotes the $j$-th predecessor of $i$ in the tree.

We notice that for each $i$, the value $c_i := \sum_{j = 0}^{depth(i)-1} (-1)^{depth(i) - depth(q^{j}(i))} e_{q^{j}(i)}$ can be computed by
$poly(n)$-size $\mathsf{AC}^0$ circuits on input, since it depends on at most $\log(n)$ values $e_j$. In this case, a $poly(n)$-size $\AC^0$ exists that computes $c_1...c_n$. 
Given these corrections, we can apply the gate $\mathsf{X}^{c_i}$ to the $i$-th vertex qubit (which is a classically-controlled $\mathsf{X}$ gate), and the resulting state is:

\begin{equation}\label{xcorr}
\ket{e_1,...,e_{n-1}} \otimes \frac{1}{\sqrt{q}}\sum_{v_1 \in \mathbb{F}_q} 
\bigotimes_{i = 1}^n \ket{(-1)^{depth(i)}v_1}.
\end{equation}

Finally, we can apply the following operation to the qudits corresponding to the vertex of odd depth:
$V\ket{m} = \ket{p - m}$ and the result state is 

\begin{equation*}
\ket{e_1,...,e_{n-1}} \otimes \frac{1}{\sqrt{q}}\sum_{v_1 \in \mathbb{F}_q} \ket{v_1}^{\otimes n}\qedhere .
\end{equation*}

\end{proof}

The reduction we previously discussed cannot be achieved using $\iQNC^0$ circuits since they cannot implement the essential corrections needed to transition from the poor man's cat states to the GHZ states. In contrast, we explore the additional computational power of $\fQAC_{\mathsf{cf}}^0$ circuits and show that they can produce the qudit GHZ states through a process involving mid-circuit measurements. As a consequence, we can show the following.

\begin{theorem}\label{quantumACcirc}
For any distinct primes $p$ and $q$, a $\fQAC_{\mathsf{cf}}^0$ circuit over qudits of dimension $q$ that can solve $\mathcal{R}_{q,p}^{q\cdot n}$ with one-sided error at most $\frac{1-\cos(\frac{2\pi}{p})}{q}$.
\end{theorem}
\begin{proof}
Combining Lemma \ref{poorconst} and Lemma \ref{reduc}, we obtain that $\fQAC_{\mathsf{cf}}^0$ circuits can create qudit GHZ states of the type $\ket{\mathsf{GHZ}_{q,n}^0}$. In addition, all the qudit $\fQNC^0$ circuits of Theorem \ref{quantumcirc} are contained in $\fQAC_{\mathsf{cf}}^0$. Therefore, all modular relation problems of the form $\mathcal{R}_{q,p}^{q\cdot n}$ fixed prime numbers $p$ and $q$ can be solved by $\fQAC_{\mathsf{cf}}^0$ circuits.
\end{proof}

\subsection{Classical lower bounds}
\label{sec:f-lower-bounds}

In this subsection, we prove the classical lower bounds for  $\mathcal{R}_{p,q}^{q\cdot n}$.  For that, we reduce it to computing $\MOD_p$ function and use the Razborov-Smolensky separations (Theorem \ref{Acseparations}) for $\AC^0[q]$. Then, we consider the 
the parallel repetitions of the previous relation problems, defined as follows.

\begin{definition}[Parallel-k modular relation problem]
The modular relation problems $k$-$\mathcal{R}_{q,p}^m:\{0,1\}^{n\cdot k}  \rightarrow \{0,1\}^{m \cdot k} $  are defined as,
\begin{equation*}
k\text{-}\mathcal{R}_{q,p}^m (x_1,...,x_k) = \Big \{ (y_1,...,y_k) \ \Big |\forall x_i \in \mathbb{F}_2^n, \ y_i \in \mathbb{F}_2^{m},\ |y_i| \pmod{q} = 0 \text{ iff } 
|x_i| \pmod{p} = 0 \Big \}.
\end{equation*} 
\end{definition}

Them, we show that $\AC^0[q]$ circuits can solve them with only a negligible small probability at best by using the Vazirani-XOR lemma.

\begin{lemma}\label{relationac}
Let $p$ and $q$ be two distinct primes. Any depth $d$ circuit in $\AC^0[q]$ that solves the $\mathcal{R}_{q,p}^{q\cdot n}$ problem with one-sided error probability at most $1/n^{\Omega (1)}$ has size $\exp(n^{1/(2d-\Theta(1)})$.
\end{lemma}
\begin{proof}
Let us suppose, for the purpose of a contradiction, that there exists a circuit $C$ in $\AC^0[q]$ of depth $d-3$ and size equal to $\exp(n^{1/(2d-\Theta(1))})$ that computes $\mathcal{R}_{q,p}^{q\cdot n}$, with one-sided error probability at most $\frac{1}{n^{o(1)}}$. We will assume here that $C$ computes the yes-instance with probability $1$, and the arguments follow analogously if $C$ computed the no-instances with probability $1$.

We will now prove that if such a $C$ exists, then there exists a $\AC^0[q]$ circuit $C'$ that computes the $\MOD_p$ function correctly with probability $1-negl(n)$. $C'$ consists of running $O(n^2)$ times the circuit $C$ in parallel and then applying a $\MOD_q$ gate to each output string, and then applying a $\mathsf{NOR}$  gate to the output of this layer of $\MOD_q$ gates. Figure \ref{fig:reductions} provides a representation of $C'$.

If $|x| \mod p=0$, after running $C$ on $x$, the resulting strings $y_{i}$ are such that $|y_{i}| \mod q=0$. Thus, the $\MOD_q$ gates to each of these strings lead to output $0$. Therefore, after the $\mathsf{NOR}$ gate, the result for this final bit will always be 1. 

On the other hand, if $|x| \mod p\neq 0$, with probability $1/n^{o(1)}$, the strings $y_i$ after running $C$ have the desired property $|y_i| \mod q \neq 0$. Therefore, with the same probability, the $\MOD_q$ gate applied to this string will result in the Boolean value $1$, and the probability of at least one of them being equal to $1$ is at least $1- negl(n)$. Therefore, the output of the  $\mathsf{NOR}$ gate is $0$ with probability $1- negl(n)$.

This leads to a depth-$d$ circuit  and size $\exp(n^{1/(2d-\Theta(1))})$, that computes  $\MOD_p$ in $\AC[q]$, contradicting Theorem \ref{Acseparations}. Thus, circuit $C$ with such properties cannot exist.
\end{proof}

\begin{figure}[h]
\begin{center}
\begin{tikzpicture}[font=\tiny]

    \tikzset{
        gate/.style={draw=black, fill=white, rectangle, align=center}
    }

    \def\spacing{0.3}

    \foreach \y/\label in {
        0/{$y_{m,n\cdot q}$},
        1.2*\spacing/{$\scalebox{0.7}{\vdots}$},
        2*\spacing/{$y_{m,2}$},
        3*\spacing/{$y_{m,1}$},
        4.8*\spacing/{$y_{2,n\cdot q}$},
        5.9*\spacing/{$\scalebox{0.7}{\vdots}$},
        6.6*\spacing/{$y_{2,2}$},
        7.4*\spacing/{$y_{2,1}$},
        9.4*\spacing/{$y_{1,n\cdot q}$},
        10.6*\spacing/{$\scalebox{0.7}{\vdots}$},
        11.4*\spacing/{$y_{1,2}$},
        12.4*\spacing/{$y_{1,1}$}
    }{
        \node[anchor=east] at (-0.5,\y) {\label};
    }

    \foreach \y/\ystart/\yend in {
        1.5*\spacing/3*\spacing/0,
        6.4*\spacing/8.1*\spacing/5*\spacing,
        11*\spacing/13*\spacing/9*\spacing
    }{
        \draw[gate] (1, \yend-0.2*\spacing) rectangle (2, \ystart+0.2*\spacing);
        \node at (1.5, \y) {$\MOD_q$};
    }

    \node at (1.5, 4.2*\spacing) {$\scalebox{0.7}{\vdots}$};
    \node at (-0.5, 4.2*\spacing) {$\scalebox{0.7}{\vdots}$};

    \foreach \y in {
        0,2*\spacing,3*\spacing,
        5*\spacing,6.6*\spacing,7.4*\spacing,
        9.4*\spacing,11.4*\spacing,12.4*\spacing
    }{
        \draw (-0.2,\y) -- (1,\y);
    }

    \foreach \y/\yend in {
        1.5*\spacing/5.8*\spacing,
        6.4*\spacing/6.6*\spacing,
        11*\spacing/7.4*\spacing
    }{
        \draw (2,\y) -- (2.7,\y) -- (2.7,\yend) -- (3.5,\yend);
    }

    \node at (3.5, 6.6*\spacing) {$\scalebox{0.7}{\vdots}$};

    \draw[gate] (3.5, 5.5*\spacing) rectangle (4.2, 7.9*\spacing);
    \node at (3.85, 6.6*\spacing) {$\mathsf{NOR}$};
    \draw (4.2,6.75*\spacing) -- (4.7,6.75*\spacing);

\end{tikzpicture}
\caption{Representation of a classical $\AC^0[q]$ circuit, parameterized by $q$ and $n$, that reduces a solution to the $\mathcal{R}_{q,2}^{q\cdot n}$ problem to a solution of $\MOD_2(x)$, where $x$ is the initial input string of size $n$.}
\label{fig:reductions}
\end{center}
\end{figure}

Subsequently, we show that polynomial-sized $\AC^0[q]$ circuits solve the parallel repetition of the modular relation problems with probability close to zero. To obtain this bound, we focus on the success of $\AC^0[q]$ circuits in computing the $\MOD_2$ function with one-sided error.  

\begin{corollary}\label{decsionsep}
The $\MOD_2$ function can not be solved by any one-sided error  $\AC^0[q]$ circuit with $q$ prime $q\neq 2$ with a probability higher than $1/2+1/n^{\Omega(1)}$.
\end{corollary}
\begin{proof}
Let's assume there exists an $\AC^0[p]$ circuit that can solve the Yes instances of the $\MOD_2$ perfectly, then we know by Lemma \ref{relationac} that the  No instances can almost be solved correctly with probability $1/n^{o(1)}$. This implies any circuit of this type solves the $\MOD_2$  function over a random input w.p. at most $1/2+1/n^{o(1)}$. 
\end{proof}

Using all these tools, we can show that the class of $\AC^0[q]$ fails to compute completely parallel modular relation problems with a one-sided error for some parameter regimes. 

\begin{lemma}\label{paraAC}
Let $q \ne 2$ be a prime and $k\in \Theta (\log(n))$. No $\AC^0[q]$ circuit can solve $k$-$\mathcal{R}_{q,2}^{q\cdot n}$ on at least $\frac{1}{2}+\frac{1-\cos(\frac{2\pi}{p})}{2q}$ fraction of the parallel instances with one-sided errors and a probability higher then $n^{-\Omega(1)}$.
\end{lemma}

\begin{proof}
Let us suppose that there exists a circuit $C \in \AC^0[q]$, where $q \neq 2$, can solve the $k$-$\mathcal{R}_{q,2}^{q\cdot n}$ with the properties of the statement. Then, we can construct a circuit $C'$ that solves the following search problem
\begin{equation}
k\text{-}\mathcal{D}_{2,2} (x_1,...,x_k)=\Big \{ (y_1',...,y_k') \ \Big |\forall x_i \in \mathbb{F}_2^n, \ y_i' \in \mathbb{F}_2,\ y_i' + |x_i| \pmod{2} = 0 \Big \},
\end{equation}
\noindent with the same success, where $C'$ is equal to $C$ with an extra layer of $\MOD_q$ gates, thereby transforming the outcome strings from $k$-$\mathcal{R}_{q,2}^{q\cdot n}$ into bits using the mapping $\MOD_q(y_i) = y_i'$.

The circuit $C'$ induces a probability distribution $\mathcal{X}$ on $\{0,1\}^k,$ corresponding on picking random $x_1,...,x_k \in \{0,1\}^n$ and outputting $k$ output bits as follows, 
\begin{align}
(|x_1|+C_1'(x_1,\hdots,x_n) \pmod{2},...,|x_k|+C_k'(x_1,\hdots,x_n) \pmod{2}).
\end{align}

Our goal now is to prove that $\mathcal{X}$ is close to a uniform distribution $\{0,1\}^k$ on half of its bits, either for all $x_i$ with even or odd Hamming weight, if we assume one-sided error. For that, we notice that we can describe for all $S\subseteq [k]$ the following expectation value as
\begin{align}
    \left|\mathbb{E}_{z \in \mathcal{X}}[(-1)^{\chi_S(z)}]\right|&= \frac{1}{2^{|S|\cdot n}}\sum_{x_1,...,x_k \in  \mathbb{F}_{2}^{n}} (-1)^{\oplus_{i \in S}(|x_i|+C_i'(x_1,\hdots,x_n)\pmod{2})}\\ &= \frac{1}{2^{|S|\cdot n}}\sum_{x_1,...,x_k \in \mathbb{F}_{2}^{n}} (-1)^{\left(\sum_{i \in S}|x_i|+\sum_{i \in S}C_i'(x_1,\hdots,x_n)\right)\pmod{2}},
\end{align}
where $\chi_S(z) = \oplus_{i\in S} z_i$.

This allows us to bind the previous value for all the subsets $S$ based on the number of times the following equation is fulfilled for uniformly selected strings $x_i$,
\begin{equation}\label{corrmod2}
 \sum_{i \in S}|x_i| \equiv \sum_{i \in S}C_i'(x_1,\hdots,x_n) \mod 2.
\end{equation}

We observe that determining the values for the outcome bits $i \in S$ such that the condition of Equation \eqref{corrmod2} is satisfied by an $\AC^0[q]$ circuit $C'$ is as hard as computing the parity function across the concatenated input strings $||_{i\in S} x_i$. This assertion is supported by the fact that a circuit, which, if capable of determining valid outcome bits under the given condition, can be reduced to one that does compute the parity of the concatenated input strings $||_{i\in S} x_i$ with a polynomial-size $\AC^0[q]$ circuit. The latter follows from the fact that polynomial size $\AC^0[q]$ circuits can effectively compute the parity of a string of length $\log(n)$. Consequently, the probability that Equation \eqref{corrmod2} is satisfied is directly related to the probability that an $\AC^0[q]$ circuit computes the parity function across the concatenated input strings $||_{i\in S} x_i$, albeit with a one-sided error. This condition is met with a maximum probability of $1/2 + n^{-\Omega(1)}/2$, as inferred from Corollary \ref{decsionsep}. Therefore, we deduce that for any characteristic function $\chi_S$, the value of $|\mathbb{E}_{x \in D}[(-1)^{\chi_S(x)}]|$ is limited by $\frac{1}{n^{\Omega(1)}}$, for all input strings of either even or odd Hamming weight.

Now, we can apply Vazirani's XOR Lemma (Lemma \ref{Xor}) to demonstrate that the distribution $\mathcal{X}$ is at most $1/n^{\Omega(1)} \cdot 2^{k/2}$ deviated in total variational distance from the uniform distribution, on those inputs. Then, we can show by using the Chernoff bound that the event $X$ described as sampling a string with at least $1/2+\frac{1-\cos(\frac{2\pi}{p})}{2q}$ of 0's from this distribution, has a probability that decreases with, 
\begin{align}
    \Pr\Big(X \Big)&<e^{-2\cdot(4/\frac{1-\cos(\frac{2\pi}{p})}{2q})^{2}(1/2)^{2}/(k^{2})} < e^{-\Omega(k^2)}\ .
\end{align}
\noindent Thus, we obtain that the distribution $\mathcal{X}$ contains the same string with probability at most, 
\begin{equation}
  n^{-\Omega(1)}\cdot2^{k/2}+  e^{-\Omega(k^2)} = n^{-\Omega(1)+k/(2\log(n))}.
\end{equation}

\noindent Finally, we conclude that the previous value is bounded by $n^{\Omega(1)}$, given that $k=\log(n)$.
\end{proof}

\subsection{Separations}
\label{sec:f-separations}
In this section, we combine the results from \Cref{sec:f-upper-bounds} and \Cref{sec:f-lower-bounds} to achieve unconditional separation between the classical and quantum circuit classes.

\begin{theorem}\label{paralsep}
For each fixed prime $q$, there exists a relation problem that cannot be solved by polynomial-size circuits $\AC^0[q]$ with success probability at least $n^{-\Omega(1)}$, whereas there exist $\fQNC^0/\mathsf{qpoly}$ and a $\fQAC_{\mathsf{cf}}^0$ circuits over qudits of dimension $q$ with a finite gate set that can solve it with probability of least $1-o(1)$.
\end{theorem}
\begin{proof}
For the case where $q = 2$, the result is derived from Theorem 6 in \cite{Watts19}. For all $q \neq 2$, we apply the lower bounds and the respective instances of the parallel modular relation problems $k$-$\mathcal{R}_{q,2}^{q\cdot n}$ discussed in \Cref{paraAC}. These instances establish that $\AC^0[q]$ circuits solve this set of problems with a probability of at most $n^{-\Omega(1)}$.

In contrast, the $\fQNC^0/\mathsf{qpoly}$ circuits, described in \Cref{quantumcirc}, can solve each instance of the parallel repetition game with a one-sided error of $\frac{1-\cos(\frac{2\pi}{p})}{q}$. Similarly, $\fQAC_{\mathsf{cf}}^0$ circuits achieve the same success probability, as detailed in \Cref{quantumACcirc}. By combining these individual solutions, we deduce that the resulting global strategies effectively solve the $k$-$\mathcal{R}_{q,2}^{q\cdot n}$ problems, as outlined in \Cref{paraAC}, with a probability of at least $1-o(1)$. This follows by considering their respective individual success probabilities and applying the Chernoff bounds.
\end{proof}

\begin{remark}
We implicitly used the fact that the $\MOD_2$ operation is not contained in any other $\AC^0[q]$ class with $q\neq 2$, which simplifies the proof since it allowed us to uniquely use Vazirani's XOR Lemma. However, we notice that the proof could follow with any other prime values and modular functions.    
\end{remark}

One drawback of the previous results is that we need to pick a different qudit dimension for each modular relation problem if we intend to use quantum circuits with finite gate sets. However, if we consider infinite-size gate sets, as in the standart definition of $\QNC^0$ and $\QAC^0$, we can show that all the previous quantum circuits can be implemented with qubits and prove the separation for $\AC^0[p]$ for all primes $p$.

\begin{theorem}\label{ssep}
For any prime $p$, $\iQAC_{\mathsf{cf}}^0$ over qubits is not contained in  $\AC^0[p]$, i.e., $\iQAC_{\mathsf{cf}}^0 \nsubseteq \AC^0[p]$.
\end{theorem}
\begin{proof}
From \cite{reck94} and \cite{Barenco95}, we know that all gates in $\fQNC^0$ circuits over qudits can be equivalently implemented using single- and two-qubit operations in constant depth. Furthermore, all quantum operations in $\fQAC^0$ circuits used to create the qudit-GHZ state in \Cref{poorconst,reduc} are controlled by a constant number of qudits, and therefore, these operations can be implemented in $\iQAC_{\mathsf{cf}}^0$ using qubit circuits alone.

Finally, we need to show that the corrections required to transform the poor man's cat states into qudit-GHZ states can be executed with $\iQAC_{\mathsf{cf}}^0$ circuits over qubits. As described in the proof of Lemma \ref{reduc}, these corrections can be carried out using a classical $\AC^0$ circuit, which is contained within $\iQAC_{\mathsf{cf}}^0$ over qubits. Consequently, all quantum circuits from \Cref{quantumACcirc,paralsep}, for any prime $p$, can be constructed using $\iQAC_{\mathsf{cf}}^0$ circuits over qubits.
\end{proof}

\begin{remark}
We note that the previous separation is achieved using what may be the weakest resource to date. This is justified because earlier separations have demonstrated the equivalence of quantum fanout and parity gates under a $\QNC^0$ reduction \cite{Green02}. However, lightcone arguments show that classical fanout cannot be translated into the classical parity gate under any $\QNC^0$ reduction.

Finally, we conjecture that $\QAC^0$cannot compute parity, and we propose that $\QAC_{\mathsf{cf}}^0$ similarly fails. This conjecture is supported by the observation that classical fanout does not seem to facilitate the unitary generation of a GHZ state, suggesting that these operations cannot be translated into one another as described in \cite{rosenthal2020}. This indicates that classical fanout and the classical parity gate may not be reducible to one another in $\QAC^0$, indicating that both are weaker computational resources than their quantum counterparts.
\end{remark}

\subsection{Classical upper bounds for modular relation problems}\label{sec:f-limitations}

In this subsection, we consider the classical upper bound on the modular relation problems, which indicates the limitations on the quantum/classical separations  using our techniques.

\begin{lemma}\label{relationUP}
For any fixed primes $q$,$p$ and fixed $k_1,k_2,m \in \mathbb{N}$,  $\mathcal{R}_{q^{k_1},p^{k_2}}^{m}$  can be solved by an $\NC^0[p]$ circuit.
\end{lemma}
\begin{proof}
We start by showing that there exists an $\NC^0[p]$ circuit that computes $\MOD_{p^{k_1}}$ for any $k_1 \in \mathbb{N}$. We use the $\AC^0[p]$ circuits described in \cite{goldreich23teaching} and argue that they can be implemented in $\NC^0[p]$. 

We prove this statement by induction on the exponent $k_1$, showing that any $\MOD_{p^{j}}$ gate can be computed if we have access to $\MOD_{p^{j-1}}$ gates. The circuit construction from \cite{goldreich23teaching} is as follows:
\begin{enumerate}
    \item For each $i \in [n]$, compute $y_i = \MOD_{p_{j-1}}(x)$.
    \item For all $i \in {2,\hdots,n}$, compute $z_i = \mathsf{AND}(y_{i-1},\neg y_i)$.
    \item Compute $b_{j} = \mathsf{AND}(\MOD_{p^{j-1}}(\sum_{i \in {2,\hdots,n}} z_i), \MOD_p(|x|))= \MOD_{p^{j}}(x)$.
\end{enumerate}
\noindent Given that the final value $b_{j} = \MOD_{p^{j}}(x)$~\cite{goldreich23teaching}, in order to show that $\MOD_{p^{j}}$ is in $\NC^0[p]$, we only need to show that each of these three steps can be implemented in $\NC^0[p]$. Let's evaluate these steps:

\begin{itemize}
    \item The first step is inherently within $\NC^0[p]$ due to inductive reasoning; given the prior case produced an $\MOD_{p^{j-1}}$ gate, only the next steps need validation.
    \item The second step can be achieved with an $\NC^0$ circuit since it solely relies on bounded fan-in $\mathsf{AND}$ gates.
    \item The third step follows since both modular gates are within $\NC^0[p]$ from our induction, combined with the $\mathsf{AND}$ gate having bounded fan-in.
\end{itemize}

\noindent Thus, any $\MOD_{p^{k_1}}$ gate can be implemented within $\NC^0[p]$ provided that $k_1$ is a constant, which determines the depth of the previously proposed circuit.

\medskip

Finally, to prove our statement, i.e., to output a valid string for the modular relation problem $\mathcal{R}_{q^{k_1},p^{k_2}}^{m}$, one needs to use the $\NC^0[p]$ circuit to compute the value of $\MOD_{p^{k_1}}(x)$ and then append it to a string like $0^{ m-1}$. We have that $|b_{k_1}0^{m-1}| \pmod{q^{k_2}} =  0$ iff
$|x| \pmod{p^{k_1}} =  0$. 
\end{proof}

Given the containment of the $\NC^0[p]$ circuit classes within the $\AC^0[p]$ classes, we have that all modular relation problems are solvable by $\TC^0$ circuits. 

\begin{corollary}
For any primes $q$,$p$ and $k_1,k_2,m \in \mathbb{N}$,  $\mathcal{R}_{q^{k_1},p^{k_2}}^{m}$  can be solved in $\TC^0$.
\end{corollary}

\section{Shallow depth quantum circuits with infinite-size gate sets}\label{quditcolapse}

This section considers constant-depth quantum circuits utilizing qubits and qudits, employing infinite gate sets and modular gates with unbounded fan-in while exploring their equivalences.

\subsection{Constant-depth qudit circuits with quantum modular operators}

In this subsection, we extend \cite{tani16} and \cite{mori2018}, and prove that the hierarchy of constant-depth quantum circuits with an infinite gate set over qudits collapses when we have access to a quantum unbounded fan-in modular gates.

We start by focusing on the quantum $\mathsf{OR}$ function, defined over Hilbert spaces of prime dimension $p$ as follows,
\begin{equation}
\mathsf{qOR} \ket{x_1,x_2, \hdots, x_n} := \ket{x_1 + \mathsf{OR}(x_2,x_3,\hdots,x_n)\pmod{p},x_2,\hdots,x_n}.
\end{equation}
\noindent with the classical $\mathsf{OR}$ operating over strings in $\mathbb{F}_p$ defined as follows\footnote{The intuition behind the $\mathsf{AND}$ and $\mathsf{OR}$ functions over strings in $\mathbb{F}_p$ is that, similar to the binary case, the $\mathsf{AND}$ function discriminates against only one string, while the $\mathsf{OR}$ function includes all strings except one.},
\begin{equation}
\mathsf{OR}(x)= 
\begin{cases}
     1 & \text{if } \sum_{i=1}^n x_i > 0, \\
     0 & \text{otherwise.}
    \end{cases}.
\end{equation}

We first claim that a qudit analog of H{\o}yer and {\v S}palek's  $\mathsf{OR}$ reduction~\cite{Hoyer05}, which allows us to reduce the problem of computing $\mathsf{qOR}$ on an arbitrary state to the evaluation of $\mathsf{qOR}$ on another state with exponentially fewer qudits.

\begin{lemma}\label{OrReduction}
Let  $p$ be a prime. There exists a $\iQNC^0[p]$  circuit $C$ over qudits of dimension $p$, such that for every $n$ qudit state $\ket{\psi}=\sum_{x\in \mathbb{F}_p^n} \alpha_x \ket{x}$, we have

\begin{equation}
    C\ket{\psi}\ket{0}^{\otimes t} = 
    \sum_{x\in \mathbb{F}_p^n,\ y\in \mathbb{F}_p^{\log_p(n)}} \alpha_x   \beta_{y,x} \ket{y}  \ket{x}\ket{0}^{\otimes t'}= \ket{\psi^*}\ket{0}^{\otimes t'},
\end{equation}
where and $\mathsf{qOR}\ket{0}\ket{\psi} = \mathsf{qOR}\ket{0}\ket{\psi^*}$ with respect to the first qudit, while the second $\mathsf{qOR}$ does operate uniquely over the first $\log(n)$ qudits of $\ket{\psi^*}$. 
\end{lemma}
\begin{proof}
We describe now our circuit $C$. In the first step of this proposed circuit, we have a layer of $n$ fanout gates, which are defined as follows
\begin{equation}
\mathsf{fanout}_p \ket{x_1, x_2, \hdots, x_n} := \ket{x_1,(x_1+x_2)\mod p,\hdots,(x_1+x_n)\mod p}.
\end{equation}
The $k$-th $\mathsf{fanout}_p$ gate is applied to the $k$-th qudit of $\ket\psi$ and a fresh $\log_p(n)$-qudit state initialized on $\ket{0}^{\otimes \log(n)}$\footnote{Note that this fanout gate can be equally created from the $\mathsf{qMOD}_p$ gate, $\mathsf{fanout}_p=(\mathsf{F}_p)^{\otimes n}\mathsf{qMOD}_p(\mathsf{F}_p)^{\otimes n}$.}.

In parallel to these $\mathsf{fanout}_p$ gates, we create $\log_p(n)$ copies of $n$-qudit GHZ states. These can be created by an $\fQNC_p^0[p]$ circuit using the method proposed in \cite{Peres23}.  We have the following state 
\begin{equation}
\ket{\psi^{(1)}}=C_{step\ 1} \ket{\psi}\ket{0}^{\otimes t}= \sum_{x\in \mathbb{F}_p^n} \alpha_x \ket{x}^{\otimes \log_p(n)}  
\ket{\mathsf{GHZ}_{p,n}^0}^{\otimes \log_p(n)} \ket{0}^{\otimes t'}.
\end{equation}

For each $1 \leq k \leq  \log_p(n)$, we apply $n$ controlled rotations, where the $j$-th rotation of the $k$-th block will have as control the $j$-th qudit of the $k$-th copy of $\ket{x}$ and its target is the $j$-th qudit of the $k$-th GHZ state.
The amplitude for all rotations in the $k$-th block is set to a fixed value of $\frac{2\pi}{p^k}$. Importantly, these operations can be implemented in constant depth due to the parallel copies of $\ket{x}$.

For a fixed $x$, this process modifies the $k$-th GHZ state to
\begin{align}
\left(\bigotimes_{i=1}^n \mathsf{R}_Z^p\bigg(\frac{2\pi x_i}{p^k}\bigg) \right)\ket{\mathsf{GHZ}_{p,n}^0}&= \sum_{x\in \mathbb{F}_p^n} \alpha_x \bigg(\frac{1}{\sqrt{p}} \sum_{j=0}^{p-1} e^{i\cdot\big(j\cdot\frac{2\pi}{p}\cdot\frac{|x|}{p^k}\big)} \ket{j}^{\otimes n}\bigg)\\ &= \sum_{x\in \mathbb{F}_p^n} \alpha_x \sum_{j=0}^{p-1} c_{j,k,x} \ket{\mathsf{GHZ_{p,n}^j}},
\end{align}
\noindent with $(c_{j,k,x})^2=e^{i\big (j\big(\frac{2\pi}{p}(1- \frac{|x|}{p^k}\big ) \big )}$.

Note that for $|x|=0$, all these GHZ states remain unchanged and are equal to $\ket{\mathsf{GHZ}_{p,n}^0}$. However, in contrast, we will now prove that for all basis states $\ket{x}$ with $|x| \geq 1$, one of the $\log_p(n)$ GHZ states we started with ends up in the form $\ket{\mathsf{GHZ}_{p,n}^m}$, where $m \neq 0$. In other words, for each nonzero $x$, there exists one value of $k$ for which $c_{j,k,x}$ equals $1$ for a single value of $j$ and $0$ for all others. Meanwhile, for the same value of $x$ and all other values of $k$, $c_{j,k,x}$ takes non-unitary values for all values of $j$. To establish this, we utilize the fact that any integer number $|x|$ can always be decomposed as
\begin{equation}
         |x|=p^{a_x}(p\cdot b_x+m_x),
\end{equation}
\noindent with $a_x,b_x\in \mathbb{N}$ and $m_x \in [1,2,\hdots,p-1]$. Hence, if $|x| > 0 $, then we have that 
\begin{equation}
\left(\bigotimes_{i=1}^n \mathsf{R}_Z^p\bigg(\frac{2\pi x_i}{p^{(a_x+1)}}\bigg)\right)\ket{\mathsf{GHZ}_{p,n}^0}=\ket{\mathsf{GHZ}_{p,n}^{m_x}},
\end{equation} \noindent with $m_x \neq 0$ by definition. Therefore, $c_{j=m_x,k=a_x,x}=1$ and  $c_{j\neq m_x,k=a_x,x}=0$, while for all the other GHZ states, these coefficients $c_{j,k\neq a_x,x}$ are non-unital as $\frac{2\pi}{p}(1- \frac{|x|}{p^k}\big )$ is not a multiple of $2\pi$. Thus, the state resulting from this set of controlled rotations is the following form,

\begin{align}
\ket{\psi^{(2)}}=&C_{step\ 2}\bigg (\sum_{x\in \mathbb{F}_p^n} \alpha_x \ket{x}^{\otimes \log_p(n)}
\ket{\mathsf{GHZ}_{p,n}^0}^{\otimes \log_p(n)}  \bigg )  \otimes \ket{0}^{\otimes t'}\\
=& \Bigg( \alpha_{0^n} \ket{0}^{\otimes n \log_p(n)} \ket{\mathsf{GHZ}_{p,n}^0}_0 \ket{\mathsf{GHZ}_{p,n}^0}_1 \hdots \ket{\mathsf{GHZ}_{p,n}^0}_{\log_p(n)} +\sum_{x\in \mathbb{F}_p^n\setminus 0^n} \alpha_x \ket{x}^{\otimes \log_p(n)}\label{eq:ghz1}\\& \hdots \left(\sum_{j=0}^{p-1} c_{j,a_x,x} \ket{\mathsf{GHZ}_{p,n}^j}\right)\ket{\mathsf{GHZ}_{p,n}^{m_x}}_{(a_x+1)}\left(\sum_{j=0}^{p-1} c_{j,(a_x+2),x} \ket{\mathsf{GHZ}_{p,n}^j}\right) \hdots \Bigg )  \otimes \ket{0}^{\otimes t'}. \label{eq:ghz2}
\end{align}

We then apply Fourier gates $\mathsf{F}_p^{\otimes n}$ to the qudits initially entangled in GHZ states. Employing Lemma \ref{projection}, we deduce that when $|x|=0$, the resultant states form a superposition of strings with a Hamming weight congruent to zero modulo $p$. Conversely, for $|x| \geq 0$, at least one state transitions into a superposition of strings congruent modulo $p$ to the inverse of $m_x$. It is important to note that since $m_x \mod p$ is invariably non-zero, its additive inverse $-m_x \mod p$ is likewise non-zero. Subsequently, we apply a $\mathsf{qMOD}_p$ gate with each of these states as control and as a target a new single qudit in a freshly initialized computational basis state $\ket{0}$ in our third register. The existence of the state described earlier ensures the $\mathsf{qOR}$'s intended effect, as at least one qudit in the third register will deterministically shift to a basis different from $\ket{0}$. For $|x| \geq 0$, since the GHZ states prior to the Fourier gates are superpositions in the Qudit-GHZ orthogonal X-basis, these strings after the Fourier gates are superpositions that align modulo $p$ with potentially any value in $\mathbb{F}_p$. Consequently, in the third register, all other qudits assume superpositions dependent on the input, yet they do not influence the $\mathsf{qOR}$ operation.

Afterward, we apply a layer of the inverses of the Fourier gates $\mathsf{F}_p^{\otimes n}$ to exactly the same qudits in the second register as before, and consequently, they revert to the same states as described in Equations \ref{eq:ghz1} and \ref{eq:ghz2}, being for all $k\in \{1,\hdots,\log(n)\}$ of the form 
\begin{equation}
 \sum_{x\in \mathbb{F}_p^n} \alpha_x \bigg(\frac{1}{\sqrt{p}} \sum_{j=0}^{p-1} e^{i\cdot\big(j\cdot\frac{2\pi}{p}\cdot\frac{|x|}{p^k}\big)} \ket{j}^{\otimes n}\bigg)= \sum_{x\in \mathbb{F}_p^n} \alpha_x \sum_{j=0}^{p-1} c_{j,k,x} \ket{\mathsf{GHZ}_{p,n}^j}.
\end{equation}

This is possible because previously, we simply applied a $\mathsf{qMOD}_p$ with these states as control and did not change the target qudits. Thus, these states return to the same bases and states after the use of the inverse Fourier gates. Therefore, after these operations, we obtain the subsequent state, for which we do not represent the qudits of the third register in its position but next to each GHZ state that generated each one of them for a simpler representation.

\begin{align}
C_{step\ 3} \ket{\psi^{(2)}}\ket{0}^{\otimes t'}  = \Bigg ( \alpha_{0^n} \ket{0}^{\otimes (n\cdot \log_p(n))}\ket{\mathsf{GHZ}_{p,n}^0}_0 \ket{\mathsf{GHZ}_{p,n}^0}_1 \hdots \ket{\mathsf{GHZ}_{p,n}^0}_{\log_p(n)}\ket{0}^{\otimes  \log_p(n)} \\  
+ \sum_{x\in \mathbb{F}_p^n\setminus 0^n} \alpha_x \ket{x}^{\otimes \log_p(n)}  \hdots \left(\sum_{j=0}^{p-1} c_{j,a_x,x} \ket{\mathsf{GHZ}_{p,n}^j}\ket{-j}\right) \ket{\mathsf{GHZ}_{p,n}^{m_x}}_{(a_x+1)} \ket{-m_x}\\
\left(\sum_{j=0}^{p-1} c_{j,(a_x+2),x} \ket{\mathsf{GHZ}_{p,n}^j}\ket{-j}\right) \hdots \Bigg) \otimes \ket{0}^{\otimes t''} .
\end{align}

At this stage, all the qudits in the three registers are entangled. However, we can disentangle the first and third registers from the second, which contains the GHZ states. This disentanglement can be achieved using the third register that retains the phase information of the equivalent of the GHZ states, creating the entanglement between these two registers. For that, we use generalized controlled rotations, with the control based on the qudits of the third register, applied to one qudit of the respective GHZ state of amplitude $\frac{2\pi}{p^k}$. For simplicity, we will consider first the effect on a single basis of $\ket{x}$,

\begin{align}
\alpha_x \ket{x}^{\otimes \log_p(n)}  \hdots  \left(\sum_{j=0}^{p-1} c_{j,a_x,x} \left(\mathsf{CGR}\left(\frac{2\pi}{p^k}\right) \ket{\mathsf{GHZ}_{p,n}^j}\ket{-j}\right)\right)\hdots  \\
=\alpha_x \ket{x}^{\otimes \log_p(n)}  \hdots  \left( \frac{1}{\sqrt{p}} \sum_{j=0}^{p-1} e^{i\cdot\big(j\cdot\frac{2\pi}{p}\cdot\frac{|x|}{p^k}\big)} \left(\mathsf{CGR}\left(\frac{2\pi}{p^k}\right) \ket{j}^{\otimes n}\ket{-j}  \right)\right)\hdots \\
=\alpha_x \ket{x}^{\otimes \log_p(n)}  \hdots  \left( \frac{1}{\sqrt{p}} \sum_{j=0}^{p-1} e^{i\cdot\big(j\cdot\frac{2\pi}{p}\cdot\frac{|x|}{p^k}\big)}e^{i\cdot\big(-j\cdot\frac{2\pi}{p}\cdot\frac{|x|}{p^k}\big)}  \ket{j}^{\otimes n}\ket{-j}  \right)\hdots \\
=\alpha_x \ket{x}^{\otimes \log_p(n)}  \hdots  \left(\sum_{j=0}^{p-1} \ket{\mathsf{GHZ}_{p,n}^0}\ket{-j}  \right)\hdots \ .
\end{align}

\noindent Now as the previous transformation is independent of the respective basis $\ket{x}$, we obtain the following state, 
\begin{align}
C_{step\ 4} \ket{\psi^{(3)}}\ket{0}^{\otimes t''}  = \Bigg ( \alpha_{0^n} \ket{0}^{\otimes (n\cdot \log_p(n))}\ket{\mathsf{GHZ}_{p,n}^0}_0 \ket{\mathsf{GHZ}_{p,n}^0}_1 \hdots \ket{\mathsf{GHZ}_{p,n}^0}_{\log_p(n)}\ket{0}^{\otimes  \log_p(n)} \\  
+ \sum_{x\in \mathbb{F}_p^n\setminus 0^n} \alpha_x \ket{x}^{\otimes \log_p(n)}  \hdots \ket{\mathsf{GHZ}_{p,n}^0}\left(\sum_{j=0}^{p-1} c_{j,a_x,x} \ket{-j}\right) \ket{\mathsf{GHZ}_{p,n}^{0}}_{(a_x+1)} \ket{-m_x}\\
\ket{\mathsf{GHZ}_{p,n}^0}\left(\sum_{j=0}^{p-1} c_{j,(a_x+2),x} \ket{-j}\right) \hdots \Bigg) \otimes \ket{0}^{\otimes t''} .
\end{align}

Consequently, we end up with $\log_p(n)$ GHZ states of the form $\ket{\mathsf{GHZ}_{p,n}^0}$, which are no longer entangled with the other registers, as intended. These states and the repetitions of the initial state can be easily removed using the fanout gates and their inverses\footnote{The inverse of the $\mathsf{fanout}_p$ gate is simply  $(\mathsf{fanout}_p)^{p-1}$.}  obtaining the following state, 

\begin{align}
\ket{\psi^{(5)}}=C_{step\ 5} \ket{\psi^{(4)}}=\Bigg ( \alpha_{0^n} \ket{0}^{\otimes (n+ \log_p(n))}+ \sum_{x\in \mathbb{F}_p^n\setminus 0^n} \alpha_x \ket{x}\hdots \left(\sum_{j=0}^{p-1} c_{j,a_x,x} \ket{-j}\right)\ket{-m_x}   \label{eq:ortho1}\\  
 \left(\sum_{j=0}^{p-1} c_{j,(a_x+2),x} \ket{-j}\right) \hdots  \Bigg) \otimes \ket{0}^{\otimes t'''} 
= \sum_{x\in \mathbb{F}_p^n,\ y\in \mathbb{F}_p^{\log_p(n)}}\alpha_x     \beta_{y,x}\ket{x}\ket{y}\otimes \ket{0}^{\otimes t'''}\label{eq:ortho2}.
\end{align}

Finally, if we take $\ket{\psi^*}$ as $\ket{\psi^{5}}$ with the first and second register swapped, we obtain that the former state satisfies the property $\mathsf{qOR}\ket{0}\ket{\psi} = \mathsf{qOR}\ket{0}\ket{\psi^*}$ on the first qudit, with the second $\mathsf{qOR}$ applied only to the first $\log(n)$ qudits. This is justified by that fact that the basis $\ket{0}^{\otimes \log_p(n)}$  in $\ket{\psi^*}$ has the same amplitude as the basis $\ket{0}^{\otimes n}$ in $\ket{\psi}$. The same principle forcefully applies to set all the bases orthogonal to the previously described, as demonstrated in Equations \ref{eq:ortho1} and \ref{eq:ortho2}. Therefore, the effect of the $\mathsf{qOR}$ function on the first qudit state is equal for both cases considered completing the proof.
\end{proof}

The second step shows an exponential-size $\iQNC^0$ circuit over qudits of dimension $p$ for $\mathsf{qOR}$.  For that, we use the following theorem on the representation of functions as a multi-linear polynomial over $\mathbb{F}_p$, expanding previous uses of these objects over $\mathbb{F}_2$ and $\mathbb{F}_3$ \cite{mori2018,Oliveira22,Mackeprang_2023}.

\begin{theorem}\label{Fourier}\cite{luong2009fourier}
Every function $f:\mathbb{F}_p^n \rightarrow \mathbb{R}$ can be expressed as a polynomial
\begin{equation}\label{decomp}
    f(x)= \sum_{k_1,k_2,\hdots,k_n \in \mathbb{F}_p} \widehat{f}(\chi_{k_1,k_2,\hdots,k_n})\cdot\chi_{k_1,k_2,\hdots,k_n}(x) ,
\end{equation}
\noindent with $\widehat{f}(\chi_{k_1,k_2,\hdots,k_n})$ being a real coefficient, 
and $\chi_{k_1,k_2,\hdots,k_n}(x)$  multi-linear function in $\mathbb{F}_p$. 
\end{theorem}

We now extend the techniques of \cite{tani16} to show that our exponential-size circuit for $\mathsf{qOR}$.

\begin{lemma}\label{qOR}
For any prime $p$, $\mathsf{qOR}$ can be implemented on an n-qudit state using $\mathcal{O}(n\cdot p^n)$ operations with a circuit over qudits of dimension $p$ in $\iQNC^0[p]$.
\end{lemma}
\begin{proof}
Let us describe the $\iQNC^0$ circuit over qudits of dimension $p$ that computes $\mathsf{qOR}$. The first layer of the circuit consists of $n-1$ fanout gates over the respective qudit dimension, where the $i$-th gate is controlled by the $(i+1)$-th qudit of the initial state to a new register initialized to $\ket{0}^{\otimes p^{n-1}}$. This maps the initial state $    \ket{0}\ket{\psi}\ket{0}^{\otimes t}$ to
\begin{align}
\ket{0}\sum_{x\in \mathbb{F}_p^{n-1}} \alpha_x \ket{x}(\ket{x_2,\hdots, x_n})^{\otimes p^{n-1}}\ket{0}^{\otimes t'}. 
\end{align}
We then apply $p^{n-1}$ parallel $\mathsf{qMOD}_p$ gates, where the $i$-th gate is applied on the $i$-th block of the previous state, and the $\mathsf{qMOD}_p$ computes a different $\chi_{k_1,k_2,\hdots,k_n}$  leading to
\begin{align}
\ket{0} \sum_{x\in \mathbb{F}_p^{n}} \alpha_x \ket{x}
\left(\bigotimes_{k_1,...,k_{n-1} \in \mathbb{F}_p} \ket{x_2,\hdots, x_n}\ket{\chi_{k_1,k_2,\hdots,k_{n-1} }(x_2,...,x_n)}\right) \ket{0}^{\otimes t'}. 
\end{align}
We then create a $\ket{\mathsf{GHZ}_{p,p^{n-1}}^0}$ state and we apply generalized controlled rotations as follows. The $i$-th rotation has as control the register where $\chi_{k_1,k_2,\hdots,k_n}(x_2,...,x_n)$ is computed and the target is the $i$-th qudit of the GHZ state.
The angle of the rotation is $\widehat{f}_{\overline{ \mathsf{OR}}}(\chi_{k_1,k_2,\hdots,k_n})$. These angles are defined by the $\overline{\mathsf{OR}}$ function which is defined as follows,
\begin{equation}
\overline{ \mathsf{OR}}(x)= 
\begin{cases}
     p-1 & \text{if } \sum_{i=1}^n x_i > 0, \\
     0 & \text{otherwise.}
    \end{cases}
\end{equation}
\noindent and that we decompose, according to \Cref{Fourier} as
\begin{equation}
    \overline{ \mathsf{OR} }(x)=  \sum_{k_1,k_2,\hdots,k_n \in \mathbb{F}_p} \widehat{f}_{\overline{ \mathsf{OR}}}(\chi_{k_1,k_2,\hdots,k_n})\cdot \chi_{k_1,k_2,\hdots,k_n}(x)  
\end{equation}
\noindent to obtain the following state on this last register, 
\begin{align*}
 \ket{\overline{ \mathsf{OR}}}&= \sum_{x\in \mathbb{F}_p^{n}} \alpha_x \bigg(\frac{1}{\sqrt{p}} \sum_{j=0}^{p-1} e^{j\cdot\frac{\pi}{p}\big (\sum_{k_1,k_2,\hdots,k_{n-1} \in \mathbb{F}_p} \widehat{f}_{\overline{ \mathsf{OR}}}(\chi_{k_1,k_2,\hdots,k_{n-1}})\cdot\chi_{k_1,k_2,\hdots,k_{n-1}}(x_2\hdots x_n) \big)} \ket{j}^{\otimes p^{n-1}}\bigg ) \\
 &= \sum_{x\in \mathbb{F}_p^{n}} \alpha_x \bigg(\frac{1}{\sqrt{p}} \sum_{j=0}^{p-1} e^{\frac{j\cdot\pi\cdot \overline{ \mathsf{OR}}(x_2\hdots x_n)}{p}} \ket{j}^{\otimes p^{n-1}} \bigg) = \sum_{x\in \mathbb{F}_p^{n}} \alpha_x \ket{\mathsf{GHZ}_{p,p^{n-1}}^{\overline{\mathsf{OR}}(x_2,\hdots,x_n)}}.
\end{align*}

In the last step, we apply the Fourier gate to all the qudits in the newly created state in the third register and a single use of the quantum modular gate $\mathsf{qMOD}_p$ from all the qudits to the first qudit in the state $\ket{0}$. Using~\Cref{fullbasis,projection}, we obtain the state

\begin{align*}
 &\sum_{x\in \mathbb{F}_p^n } \alpha_x \ket{\overline{\mathsf{OR}}(x_2\hdots x_n)^{-1}}\ket{ x_1,x_2,\hdots,x_n}\\  &\left(\bigotimes_{k_1,...,k_{n-1}} \ket{x_2,\hdots, x_n}\ket{\chi_{k_1,k_2,\hdots,k_{n-1} }(x_2,...,x_n)}\right)\ \Bigg( \frac{1}{\sqrt{p^{n-2}}} \sum_{\substack{z \in \mathbb{F}_p^{n-1} \\ |z| \mod p=-\overline{ \mathsf{OR}}(x_2,\hdots,x_n) }} \ket{z} \Bigg ).
\end{align*}

Next, we apply a $\mathsf{SUM}$ gate to the first qudit, with the first qudit from the initial state $\ket{\psi}$. This ensures we obtain the expected output state from applying $\mathsf{qOR}$ function on the input state $\ket{\psi}$ given that $-\overline{\mathsf{OR}}(x_2,\hdots,x_n) =\mathsf{OR}(x_2,\hdots,x_n)$. Finally, we execute the inverse operations on the auxiliary qudits to decouple these from the input state. We can do this because all operations acting on the input state and auxiliary qudits are unitary. In addition, the inverse operations of $\mathsf{fanout}_p$, $\mathsf{qMod}_p$, and of the controlled rotations—are all within $\fQNC^0[p]$ circuits over qudits of dimension $p$. See Section \ref{operations} for details on the first two operations.
\end{proof}

We now prove our first collapse of the different classes on constant-depth quantum circuits. 

\begin{theorem}\label{lemma:i-collapse-qac}
For any $p$ prime, $\iQNC^0[p]$ and $\iQAC^0[p]$ over qudits of dimension $p$ are equivalent, i.e. $\iQNC^0[p]=\iQAC^0[p]$.
\end{theorem}
\begin{proof}
To prove this collapse, we need only to show that $\mathsf{qOR}$ and the  $\mathsf{qAND}$  operation can be implemented by an $\iQNC^0[p]$ circuit over qudits of dimension $p$. The latter gate is already the respective Toffoli $\mathsf{X}$ gate, which we can use to obtain any other multi-qudit Toffoli gate. 

To implement $\mathsf{qOR}$, we integrate circuits $C_1$ and $C_2$ from Lemmas \ref{OrReduction} and \ref{qOR}, respectively. Circuit $C_1$ effectively reduces the number of qudits in the state $\ket{\psi}$, generating a new state $\ket{\psi^*}$ that preserves the $\mathsf{qOR}$ output alongside a fresh $\ket{0}$, and $C_2$ performs the actual $\mathsf{qOR}$ computation. By applying $C_1$ to all qudits of $\ket{\psi}$ except the first, and subsequently feeding $\ket{0}\ket{\psi^*}$ to $C_2$. The global output is generated by applying the $\mathsf{SUM}$ gate to the first qudit $\sum_{x \in [p^{n-1}]} \alpha_{x_2,\hdots,x_n} \ket{\mathsf{OR}(x_2,\hdots,x_n)}$ resulting from $C_2$ and the first qudit of the initial $\ket{\psi}$ state.

To conclude the proof, we return all auxiliary qudits to their initial state, $\ket{0}^{\otimes t}$. Initially, we observe that $\log(n)$ auxiliary qudits arise from the $\mathsf{OR}$ reduction. We then apply the exponentially large $\mathsf{qOR}$ operator, described in \Cref{qOR}, to these $\log_p(n)$ qudits. Importantly, this step does not introduce additional auxiliary qudits. Thus, our goal now is to reverse the state of these $\log(n)$ auxiliary qudits to $\ket{0}^{\otimes \log_p(n)}$, without impacting the result obtained from the $\mathsf{qOR}$ operation. We notice that undoing the operation on the $\log(n)$ qubits is possible since the operations of \Cref{OrReduction} can be reverted even if the first and second registers are entangled with the qubit where the $\mathsf{qOR}$ was applied. This implies that the $\log(n)$ qudits were unnecessary to achieve the intended final state. Hence, we can reversibly return these qudits to their initial state without losing any information.

The reversal process is straightforward. We follow the steps used to generate these $\log_p(n)$ qudits but omit the application of the $\mathsf{qMOD}$ operator, as delineated in \Cref{OrReduction}. Instead, we employ the $\mathsf{qMOD}^{p-1}$ operator. This ensures that all qudits are restored to the $\ket{0}$ state, completing the process unitarily.
\end{proof}

For the next collapse, we need to demonstrate that the quantum threshold gate is in $\iQNC^0[p]$. For that, we first demonstrate that the quantum version of the $\mathsf{Exact}_k$ function is contained in the same circuit class. 
\begin{definition}
\noindent $\mathbf{\mathsf{Exact}}_k$. The $\mathsf{Exact}_k$ function is defined as follows,
\begin{equation}
\mathsf{Exact}_k(x) = 
\begin{cases} 
1 & \text{if } \sum_{i=1}^{n} x_i = k \\
0 & \text{otherwise}
\end{cases},
\end{equation} 
\noindent and its quantum analog as the following operation over Hilbert spaces of prime dimension $p$
\begin{equation}
\mathsf{qExact}_k \ket{x_1, x_2, \hdots, x_n} := \ket{ x_1 + \mathsf{Exact}_k(x_2,x_3,\hdots,x_n) \pmod{p},x_2,\hdots,x_n}.
\end{equation}
\end{definition}

\begin{lemma}\label{Exact}
For any prime $p$,  $\mathsf{qExact}$ can be implemented with a $\iQNC^0[p]$ circuit over qudits of dimension $p$.
\end{lemma}
\begin{proof}
We approach this problem once again using a constructive method, retaining most elements from the circuit used for the $\mathsf{qOR}$ function. The $\mathsf{qOR}$ function fundamentally discerns whether a computational basis state has a Hamming weight of zero. However, the initial step involving the $\mathsf{qOR}$ reduction will be modified. Previously, this reduction applied controlled rotations to GHZ states, effectively implementing a rotation with an amplitude proportional to the Hamming weight of the computational basis states, denoted as $\frac{2\pi |x|}{p^i}$. The new circuit introduces additional rotations to each $i$-th GHZ state with an amplitude of $-\frac{2\pi k}{p^i}$ respectively. This supplemental rotation emulates the behavior of the original circuit, equating the effects of a string with a Hamming weight of zero to those of a string with a Hamming weight of $k$. As a result, when these modified states are subsequently processed by the remaining $\mathsf{qOR}$ circuitry, the outcome will precisely compute the $\mathsf{qExact}_k$ function.
\end{proof}

Using \Cref{Exact}, we push the collapse of \Cref{lemma:i-collapse-qac} further to constant-depth quantum circuits with threshold gates.

\begin{theorem}\label{fullcollapse}
For any $p$ prime, $\iQNC^0[p]$ and $\iQTC^0$ over qudits of dimension $p$ are equivalent, i.e. $\iQNC^0[p]=\iQTC^0$.
\end{theorem}

\begin{proof}
To prove this statement, it is sufficient to show that $\mathsf{qTH}_k$ can be implemented with circuits over qudits of dimension $p$ in $\iQNC[p]$. The circuit for $\mathsf{qTH}_k$ is composed of $n-k$ $\mathsf{qExact}_t$ operations (described in Lemma \ref{Exact}) for which $t$ will range from $k$ to $n$. In this setup, the fanout gate is employed to parallelize the application of the $\mathsf{qExact}_t$ operations. As a result, we generate states of the form $\sum_{x\in \mathbb{F}_p^{n-1}} \alpha_x \ket{\mathsf{Exact}_k(x_2\hdots x_n)}\ket{x}$ based on the input state $\ket{\psi}=\sum_{x\in \mathbb{F}_p^n} \alpha_x \ket{0}^{\otimes (n-k)}\ket{x}$. 

Subsequently, the $\mathsf{qOR}$ gate will be applied to all these states, and a computational basis state $\ket{0}$, obtained with the following state,
\begin{align}
   \sum_{x\in \mathbb{F}_p^n} \alpha_x  &\ket{\mathsf{OR}(\mathsf{Exact}_k(x_2\hdots x_n),\mathsf{Exact}_{k+1}(x_2\hdots x_n),\hdots,\mathsf{Exact}_n(x_2\hdots x_n))}
   \\ &\ket{\mathsf{Exact}_k(x_2\hdots x_n)}_k \hdots \ket{\mathsf{Exact}_k(x_2\hdots x_n)}_n \ket{x}^{\otimes (n-k)}\\
   = \sum_{x\in \mathbb{F}_p^n} \alpha_x &\ket{\mathsf{TH}_k(x_1\hdots x_n)}\hdots \ket{x}^{\otimes (n-k)} \ .
\end{align}
It requires only applying the $\mathsf{SUM}$ gate between the first qudit and the first qudit of the state $\ket{\psi}$, and we obtain the state resulting from applying the $\mathsf{qTH}_k$ on the input state $\ket{\psi}$. The last step resumes applying all the inverse unitaries on the auxiliary qudits to detangle those from the resulting state. Lastly, we apply the inverse unitaries to all auxiliary qudits, disentangling them from the final state while preserving the desired output.
\end{proof}

\subsection{Constant-depth qudit circuits  with classical modular operators}

In this section, we consider the collapse of constant-depth circuits with infinite-size quantum gate sets and classical unbounded fan-in modular gates in combination with unbounded classical fanout.

We first show that quantum modular gates can be implemented with constant-depth measurement patterns, along with classical unbounded fan-in modular gates and classical fanout.

\begin{lemma}\label{cqMOD}
The $\mathsf{qMOD}_p$ gate can be implemented in $\iQNC_{\mathsf{cf}}^0[p]_\mathsf{c}$. 
\end{lemma}

\begin{proof}
Our first step in this proof is to show  that a circuit in $\iQNC_{\mathsf{cf}}^0[p]_\mathsf{c}$ can implement the $\mathsf{fanout}_p$ gate over an arbitrary quantum state, $\ket\psi=\sum_{x\in\{0,1,\hdots,p-1\}^n} \alpha_x \ket{x}$. 

For that, we first create an $\ket{\mathsf{GHZ}_{p,n}^0}$ state, which can be constructed using the circuit described in \cref{poorconst}  combined with the circuit described in \cref{reduc}  which are both contained in $\iQNC_{\mathsf{cf}}^0[p]_\mathsf{c}$. We then apply a $\mathsf{SUM}$ gate where the control-qudit is the first qudit of $\ket\psi$ and the target is the first qudit of the GHZ state. This leads to the state, 
\begin{equation}
\ket{\psi^{(1)}}= \sum_{x\in \mathbb{F}_p^n }  \alpha_x  (\mathsf{I}^{\otimes n} \otimes (\mathsf{X}_p)^{x_1} \otimes \mathsf{I}^{\otimes (n-1)} )  \ket{x}  \ket{\mathsf{GHZ}_p^0}\ .
\end{equation} 
\noindent Subsequently, we measure the first qudit of the  second register with outcome $m_1 \in \{0,1,\hdots,p-1\}$. The post-measured state will then be 
\begin{equation}
\ket{\psi^{(2)}}= \sum_{x\in \mathbb{F}_p^{n}} \alpha_{x} \ket{x_1,x_2\hdots x_n,m_1,(x_1+m_1) \mod p,\hdots,(x_1+m_1) \mod p}\ .
\end{equation}
We then apply $(\mathsf{X}_p)^{-m_1}$ to each of the  $n-1$ qudits of the second register, leading to the state
\begin{equation}
\ket{\psi^{(3)}} = \sum_{x\in \mathbb{F}_p^{n}} \alpha_{x} \ket{x_1,x_2, \hdots, x_n,m_1,x_1,\hdots,x_1}\ .
\end{equation}

Using this state, we can apply a layer of $n-1$ parallel $\mathsf{SUM}$ gates where the control gate is one of the $n-1$ last qudits, and the target is one of the qudits from the $2$nd to the $n$-th qudit in the first register. Resulting in the following state
\begin{equation}
    \ket{\psi^{(4)}}=  \sum_{x\in \mathbb{F}_p^{n}} \alpha_{x} \ket{x_1,(x_2+x_1)\ \mathsf{mod}\ p, \hdots,(x_n+x_1)\ \mathsf{mod}\ p,m_1,x_1,\hdots,x_1}\ ,
\end{equation}
which is equivalent to applying the $\mathsf{fanout}_p$ gate on state $\ket{\psi^{(3)}}$.

The last step is to remove the entanglement between the first $n$ qudits and the last $n-1$ ones. For that, we apply the Fourier gate to the last $n-1$ qudits,
and we have  
\begin{align}
    \ket{\psi^{(5)}}=  &\frac{1}{\sqrt{p^n}}\sum_{m_2\hdots m_n\in \mathbb{F}_p^{n-1}} \sum_{x\in \mathbb{F}_p^{n}} \alpha_{x}\cdot e^{i\frac{-2\pi x_1 \left( \sum_{i=2}^n m_i\right)}{p} } \\  &\ket{x_1,(x_2+x_1)\ \mathsf{mod}\ p,\hdots,(x_{n-1}+x_1)\ \mathsf{mod}\ p, m_1,\hdots,m_n}\ .
\end{align}
We can then measure the last $n-1$ qudits with outcome $m_2,...,m_n \in \mathbb{F}_p$ to determine the value of $\phi=-\sum_{i=2}^n m_i \pmod{p}$ and apply a controlled generalized rotation $\mathsf{CGR}_Z^p$ with control on $x_1$ and parameterized by $\phi$ to remove the relative phase introduced by the Fourier gates\footnote{The value of $\phi$ is classical and can be determined using the classical $\MOD_p$ gates with a binary outcome in the form of $p$ bits, with only one of them being equal to $1$, determining the value $\phi$. }. Tracing out the measured states, we have the state
\begin{equation}
\ket{\psi^{(6)}}= \sum_{x\in \mathbb{F}_p^{n}} \alpha_{x} \ket{x_1,(x_2+x_1)\ \mathsf{mod}\ p, \hdots,(x_n+x_1)\ \mathsf{mod}\ p}\ .
\end{equation}
\noindent which corresponds to state $(\mathsf{fanout_p}\ket{\psi})$, as intended. 

Finally, we only need to use the circuit that translates the $\mathsf{fanout}_p$ gate to the $\mathsf{qMOD}_p$ gate, using simply the standard translation presented by qubits \cite{Hoyer05} in higher dimensions with the Fourier gates. This and the fact that all the sub-processes described previously are contained in $\iQNC_{\mathsf{cf}}^0[p]_\mathsf{c}$ completes the proof. 
\end{proof}

Using this lemma and the results in the previous subsection we can prove the following.

\begin{lemma}\label{classical_collapse}
For any prime $p$, $\iQNC_{\mathsf{cf}}^0[p]_\mathsf{c}$ and $\iQTC^0$ over qudits of dimension $p$ are equivalent, i.e. $\iQNC_{\mathsf{cf}}^0[p]_\mathsf{c}=\iQTC^0$.
\end{lemma}
\begin{proof}
It follows from \Cref{cqMOD} that $\iQNC_{\mathsf{cf}}^0[p]_\mathsf{c}=\iQNC^0[p]$. The result follows since $\iQTC^0=\iQNC^0[p]$ circuits over qudits of dimension $p$ from \Cref{fullcollapse}.
\end{proof}

\Cref{classical_collapse} shows that for the same collapses, one only needs constant-depth quantum circuits with bounded fan-in gates, while the additional unbounded fan-in gates are classical. This was previously known for the qubit case through realizations with the MBQC model over qubits \cite{Browne11}; however, we demonstrate that this applies equally to constant-depth MBQC over prime-dimensional qudits \cite{booth22}, as they are capable of the same type of circuit collapses. We now show how these circuit classes compare to each other.

\begin{theorem}\label{qubitcol}
For any $p$ and $q$ prime, $\iQTC^0$ circuits over qudits of dimension $p$ are contained within $\iQNC_{\mathsf{cf}}^0[q]_\mathsf{c}$ circuits over qubits with abritary prime modular unbounded fan-in gates.
\end{theorem}
\begin{proof}
The proof proceeds in two steps. Firstly, we show that $\iQNC^0$ circuit over qudits of any prime dimension $p$ and $q$ are contained within $\iQNC^0$ circuits over qubits. To establish this, we need a mapping from the qudit basis to the qubit basis for both prime dimensions, $p$ and $q$. For example, qutrits can be mapped to qubits using the following correspondence: $\ket{0}\rightarrow \ket{00}, \ket{1}\rightarrow \ket{01}$ and $\ket{2}\rightarrow \ket{11}$, utilizing only a subspace of the Hilbert space of dimension 4. Consequently, any qudit can be translated into $2 \lceil p/2 \rceil$  qubits. Additionally, any unitary operator over a Hilbert space of dimensions $p$ and $2p$ can be mapped to a unitary operator over a Hilbert space of dimensions $2 \lceil p/2 \rceil$  and $4 \lceil p/2 \rceil$. This new operator applies the same transformation over the encoded qubit basis as the original did over the qudits, while the basis ignored by the encoding is operated by the identity locally.

This argument can be augmented with the demonstration that all such fixed sizes can be constructed in constant depth using single and two-qubit gates, as described in \cite{reck94} and \cite{Barenco95}. These methods prove that any unitary operator on $k$ qubits can be synthesized using at most $\mathcal{O}(k^3 4^k)$ two-qubit gates. This further suggests that all operations in $\iQNC^0$  over qudits of dimension $p$ and $q$ can be executed by an $\iQNC^0$ circuits over qubits with a specific encoding for each prime dimension.

Given that all operations in $\iQNC^0$ over qudits of dimension $q$ can be achieved within $\iQNC^0$ circuits over qubits, it follows that all operations in $\iQNC_{\mathsf{cf}}^0[q]_\mathsf{c}$ oved qudits of dimension $q$ also belong to $\iQNC_{\mathsf{cf}}^0[q]_\mathsf{c}$ circuit over qubits. This implies that $\iQNC_{\mathsf{cf}}^0[q]_\mathsf{c}$ over qubits contains the $\MOD_p$  gate. Coupled with the fact that $\iQNC^0$ circuits over qubits can also implement all the operations in $\iQNC^0$ circuits over qudits do dimension $p$, we deduce that $\iQNC_{\mathsf{cf}}^0[q]_\mathsf{c}$ circuits over qubits can execute all operations in $\iQNC_{\mathsf{cf}}^0[p]_\mathsf{c}$ for any given prime $p$. The final piece to consider is Lemma \ref{classical_collapse}, which allows us to arrive at the desired result.
\end{proof}

Building on this result, we can also extend the quantum-classical separations illustrated by \cite{Green02} and \cite{tani16}, which separated $\iQNC^0[2]$ from the classical $\AC^0[p]$ classes.

\begin{corollary}\label{sepinfinit}
For all $p$ and $q$ prime,  $\iQNC_{\mathsf{cf}}^0[q]_\mathsf{c} \nsubseteq \AC^0[p]$.
\end{corollary}
\begin{proof}
The $\iQNC_{\mathsf{cf}}^0[q]_\mathsf{c}$ circuits over qubits contains all the operations of $\iQNC_{\mathsf{cf}}^0$ circuits over qudits of arbitrary prime dimensions; therefore, it contains all the operations in $\iQTC^0$, which include all the $\MOD_k$ operations, with $k$ prime. This is combined with Razborov-Smolensky separations (Theorem \ref{Acseparations}), which states that all $\AC^0[p]$ classes fail to compute any $\MOD_k$ gate other than the $\MOD_{p}$ gate with $p$ and $k$ being two distinct prime numbers finishes the proof.
\end{proof}

\section{Conclusion and discussion}

Intending to explore potential quantum advantages in the shallow-depth regime and gain a deeper understanding of quantum circuit classes in this context, we have shown that the separations previously identified between $\QNC^0$ circuits with quantum advice and $\AC^0[2]$ circuits generalize to arbitrary primes when extending quantum circuits from qubits to qudits. Specifically, this allows us to demonstrate that the same qudit-based quantum circuits can be used in proving techniques to separate $\QAC^0$ with classical fanout from all $\AC^0[p]$ circuit classes simultaneously. The inclusion of classical fanout alone represents the potentially weakest resource known to achieve such a separation. However, it remains an open question whether a separation from the $\AC^0[p]$ classes is possible using only $\QAC^0$ or $\QNC^0$ circuits.

Additionally, all previously identified separations can be implemented using quantum circuits with a fixed gate set, thus opening the possibility for an error-corrected realization, as has been shown feasible in some quantum-classical shallow-depth separations \cite{bravyi20,hsieh2024unconditionally,caha2023colossal}. However, the question of their resilience to noise remains open.

Furthermore, when considering these circuit classes with infinite-size gate sets and additional modular gates, we have shown that the previously established collapses \cite{tani16} involving qubit circuits and the parity gate also extend to quantum circuits with arbitrary prime qudit dimensions, incorporating modular gates of the corresponding prime arity. Additionally, we have demonstrated that the same separation can be achieved using qubit circuits that include an additional modular gate of arbitrary prime arity. This result contrasts sharply with the classical setting, where $\NC^0$ circuits with access to modular gates of different primes generate distinct, separable circuit classes.

More importantly, for hardware implementations and the execution of subroutines in quantum algorithms with conjectured exponential advantage, these new collapses—including those achieved with classical modular gates—extend and may facilitate the realization of such quantum circuits. Classical multi-bit gates are technically simpler to implement than multi-qudit gates, making these results of greater practical interest.

\bibliographystyle{plain}
\bibliography{library}

\end{document}